\documentclass{amsart}

\usepackage[T1]{fontenc}
\usepackage{times}

\usepackage{tikz}
\usetikzlibrary{shapes,arrows,matrix,patterns,positioning}
\usepackage[margin=1.2in]{geometry}
\usepackage{bm}
\usepackage{color}
\usepackage{algorithm,amssymb,amsmath}
\usepackage{algorithmicx}
\newtheorem{theorem}{Theorem}

\usepackage{hyperref}
\usepackage{booktabs}
\usepackage{multirow}
\usepackage{array}

\newtheorem{defn}[theorem]{Definition}

\DeclareMathOperator{\argmin}{arg min}

\DeclareMathOperator{\mass}{mass}

\newcommand{\grad}{\nabla}
\newcommand{\RR}{\mathbb{R}}
\newcommand{\NN}{\mathbb{N}}
\newcommand{\ZZ}{\mathbb{Z}}
\newcommand{\TT}{\mathrm{T}}

\newcommand{\CC}{\mathbb{C}}

\newcommand{\cals}{\mathbf{\rho}}

\newcommand{\ud}{\,\mathrm{d}}

\newcommand{\mc}[1]{\mathcal{#1}}

\newcommand{\eps}{\epsilon}

\newcommand{\abs}[1]{\lvert#1\rvert}

\renewcommand{\Re}{\mathfrak{Re}}

\usepackage{mathtools}

\DeclarePairedDelimiter\floor{\lfloor}{\rfloor}

\title[Phase Space Sketching for Crystal Image Analysis]
{Phase Space Sketching for Crystal Image Analysis\\ based
    on Synchrosqueezed Transforms}

\author{Jianfeng Lu} \address{Department of Mathematics, Department of
  Physics, and Department of Chemistry, Duke University, Box 90320, Durham NC 27708, USA}
\email{jianfeng@math.duke.edu}

\author{Haizhao Yang}
\address{Department of Mathematics, National University of Singapore, 10 Lower Kent Ridge Road, 119076, Singapore}
\email{matyh@nus.edu.sg}

\date{\today} \thanks{J.L. was supported in part by the National Science Foundation under awards DMS-1454939 and ACI-1450372. H.Y. thanks the support of the startup grant from National University of Singapore. We would like to thank Yilong Han and Yang Xiang for helpful discussions.}

\begin{document}

\begin{abstract}
    Recent developments of imaging techniques enable researchers to
  visualize materials at the atomic resolution to better understand
  the microscopic structures of materials.  This paper aims at
  automatic and quantitative characterization of potentially
  complicated microscopic crystal images, providing feedback to tweak
  theories and improve synthesis in materials science. As such, an
  efficient phase-space sketching method is proposed to encode
  microscopic crystal images in a translation, rotation, illumination,
  and scale invariant representation, which is also stable with
  respect to small deformations.  Based on the phase-space sketching,
  we generalize our previous analysis framework for crystal images
  with simple structures to those with complicated geometry.
\end{abstract}



\maketitle

\section{Introduction}
\label{sec:intro}

Crystal image analysis at the atomic resolution has become an
important research direction in materials science recently
\cite{Berkels:08, Strekalovskiy:11, ElseyWirth:MMS, YangLuYing:2015,
  LuWirthYang:2016}. The advancement of image acquisition techniques
enable researchers to visualize materials at atomic resolution, with
images of clearly visible individual atoms and their types (see Figure
\ref{fig:ex0} (a)) and defects such as dislocations and grain
boundaries (see Figure \ref{fig:ex0} (f)). These high-resolution
images provide unprecedented opportunities to characterize and study
the structure of materials at the microscopic level, which is crucial
for designing new materials with functional properties.

The recognition of important quantities and active mechanisms (e.g.,
dislocations, grain boundaries, grain orientation, deformation,
cracks) in a material requires the use of automatic and quantitative
analysis by computers. Due to the extraordinarily large volume of
measurements and simulations in daily research activities (see Figure
\ref{fig:ex0}), in particular, in the case of analyzing a time series
of crystal images during the dynamic evolution of crystallization
\cite{crystalization1}, crystal melting \cite{melting2,melting5},
solid-solid phase transition \cite{Peng:2014}, and self-assembly
\cite{assembly1} etc., it is impractical to analyze these images
manually. In the case of crystal images with complicated geometry (see
Figure \ref{fig:ex0} (b) for an example), it becomes difficult to
recognize and parametrize the image patterns by visual
inspection. Moreover, it is also difficult, if not impossible, to
measure crystal deformation directly from crystal images by
hand. Therefore, there is a dire need for efficient tools to classify
and analyze atomic resolution crystal images automatically and
quantitatively, with minimal human intervention.

 \begin{figure}
  \begin{center}
  \begin{tabular}{ccc}
   \includegraphics[height=0.65in]{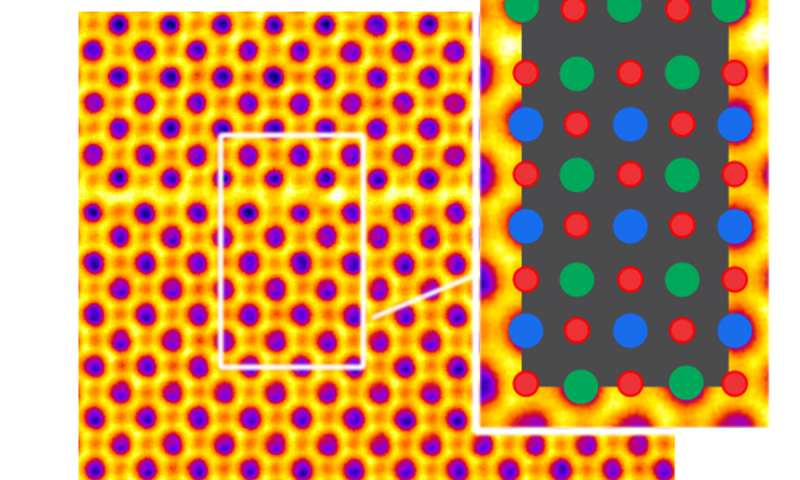}  &   \includegraphics[height=0.65in]{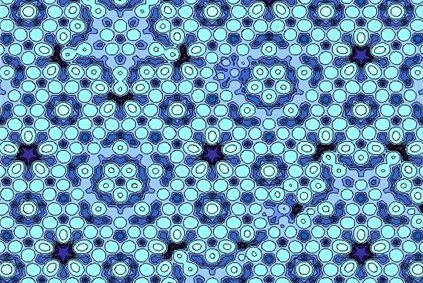} &   \includegraphics[height=0.65in]{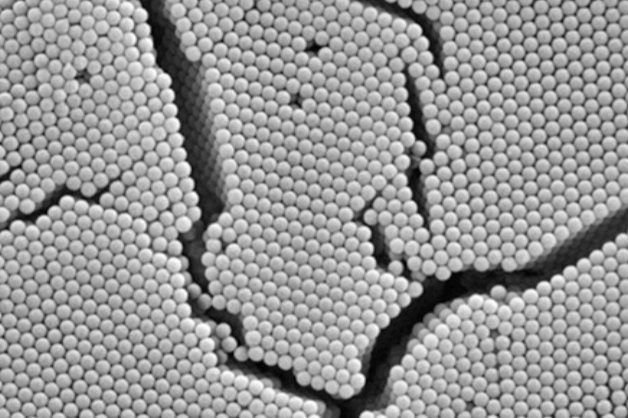}\\
   (a) &(b)&(c)\\
    \includegraphics[height=0.65in]{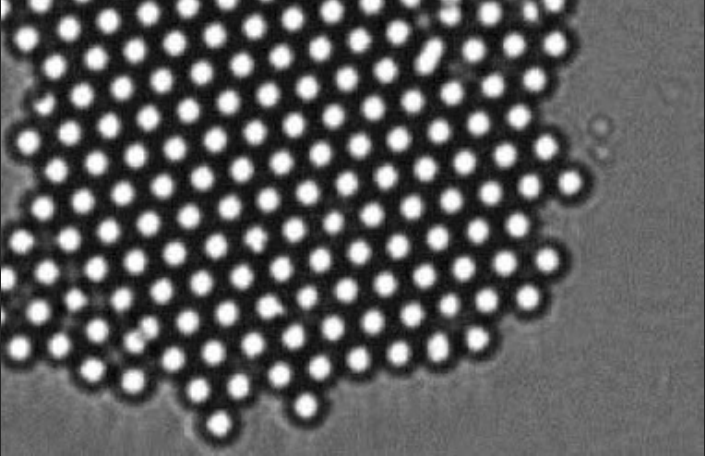} &\includegraphics[height=0.65in]{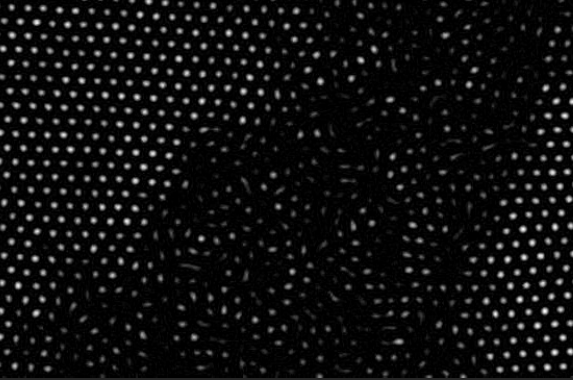}& \includegraphics[height=0.65in]{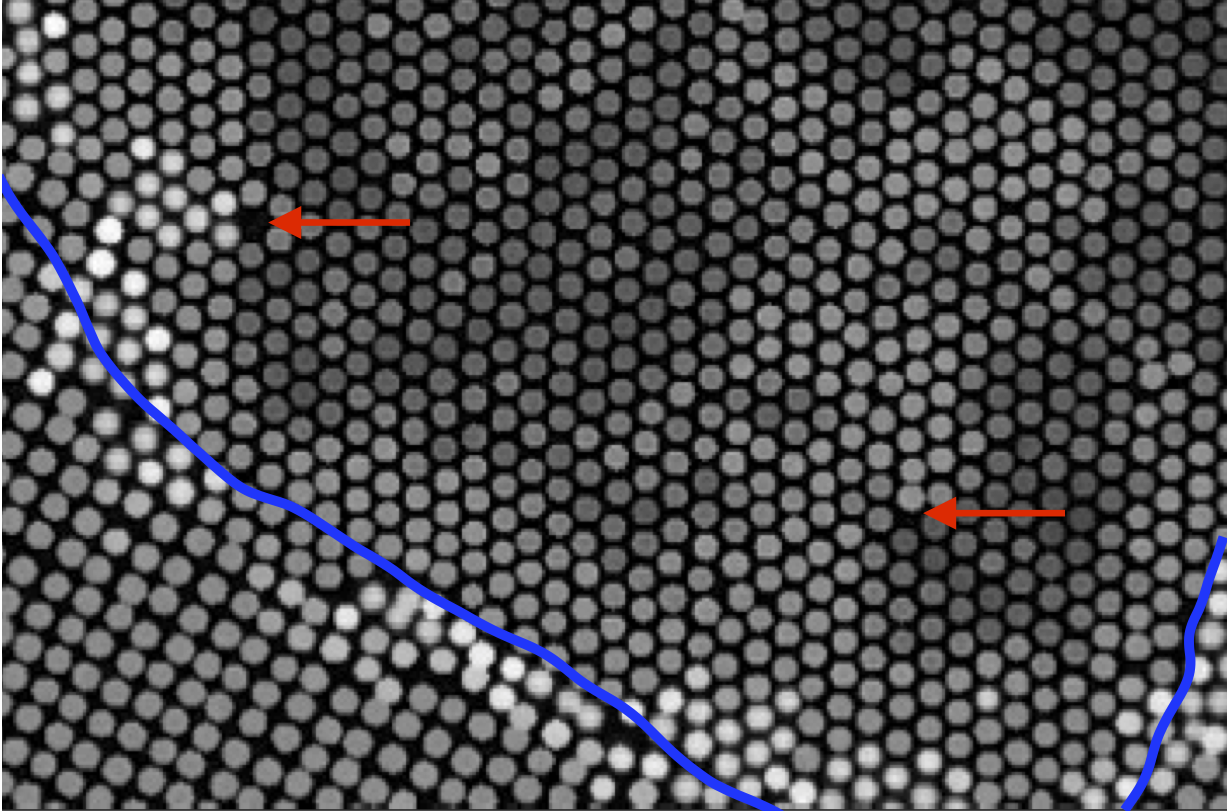}\\
    (d)&(e)&(f)
     \end{tabular}
  \end{center}
  \caption{(a) A colorized sub-Angstrom scanning transmission electron microscope image clearly shows individual atomic columns of strontium (green), titanium (blue), and oxygen (red) \cite{subatom}.  Courtesy of Greg Stone of Pennsylvania State University. (b) Potential energy surface for silver depositing on an aluminium-palladium-manganese (Al-Pd-Mn) quasicrystal surface. Similar to Figure 6 in \cite{PhysRevB.75.064205}. 
    (c) Colloidal crystal with a crack in self-assembly
    \cite{assembly1}. Courtesy of Aizenberg et al. at Harvard
    University. (d) An optical microscopy image of electric field
    mediated colloidal crystallization in fluid
    \cite{Juarez:2012}. Courtesy of Edwards et al.. (e) Melting
    behaviors of thin crystalline films \cite{melting5}. Courtesy of
    Yilong Han of Hong Kong University of Science and Technology, and
    Arjun Yodh of Pennsylvania State University. (f) A microscopic
    image showing the simulation of nucleation mechanism in
    solid-solid phase transitions \cite{Peng:2014}. Courtesy of Yilong
    Han of Hong Kong University of Science and Technology, and Arjun
    Yodh of Pennsylvania State University. Examples of dislocations
    are pointed out by red arrows and grain boundaries are indicated
    in blue.}
  \label{fig:ex0}
\end{figure}

There have been several types of methods for atomic resolution crystal
image analysis, assuming simple crystal patterns are known a priori,
typically a hexagonal or a square reference lattice in two dimensions
(see Figure \ref{fig:Bravais}). One class of methods tries to estimate
atom positions first (or assume knowledge of the atom positions), and
then compute the local lattice orientation and deformation as well as
defects via identifying the nearest neighbors of each atom
\cite{StukowskiAlbe2:10}. Other methods are based on a local,
direction sensitive frequency analysis, e.g. wavelets
\cite{SingerSinger:06} to segment the crystal image into several
crystal grains and identify their orientations. Another more advanced
class of methods formulates the crystal analysis problem (such as
segmentation) as an optimization problem with fidelity term specially
designed for the local periodic structure of the crystal image
\cite{Berkels:08,Strekalovskiy:11,ElseyWirth:MMS,7477729}.

\begin{figure}[ht!]
  \begin{center}
    \begin{tabular}{c}
      \includegraphics[height=2in]{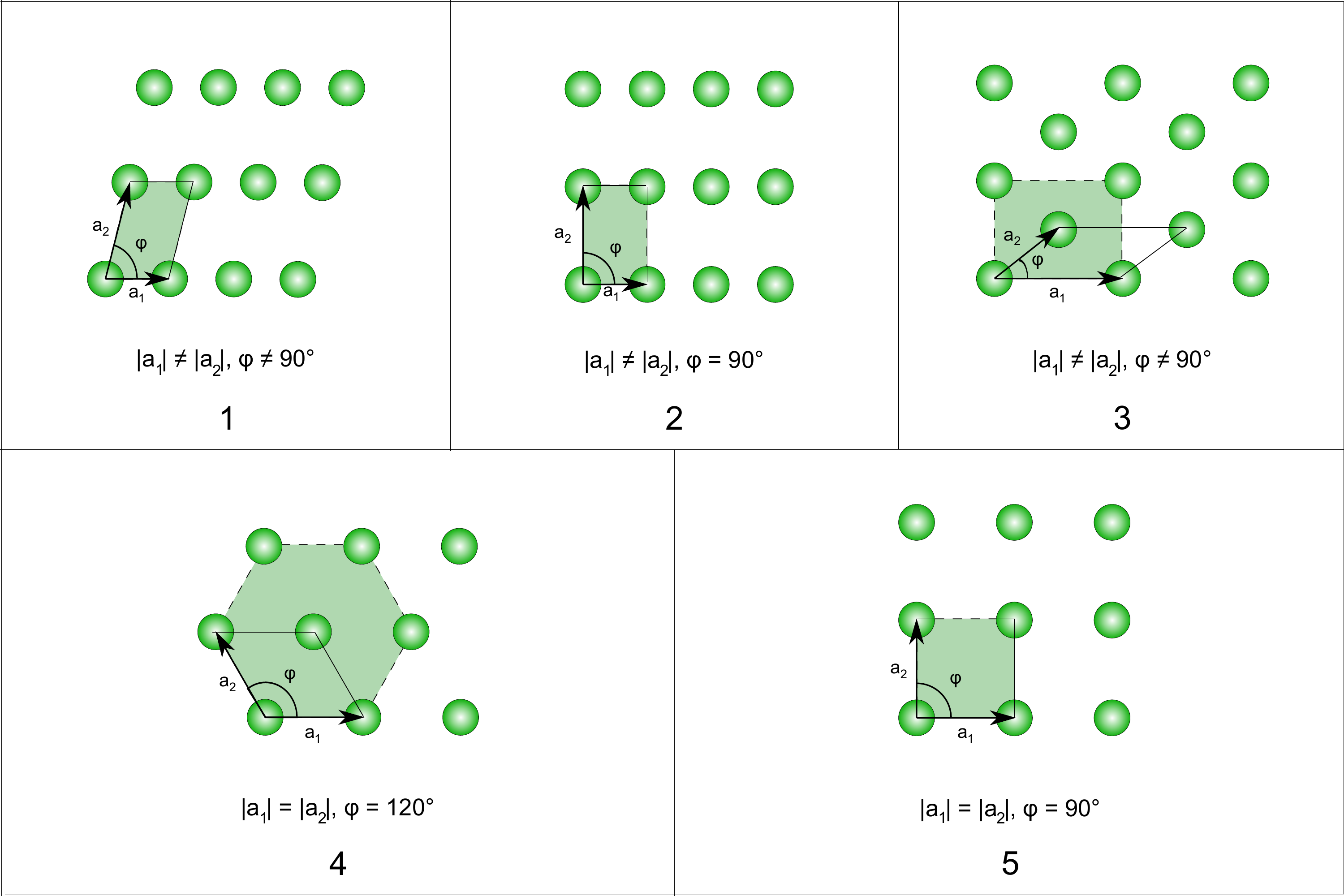}
    \end{tabular}
  \end{center}
  \caption{Five fundamental $2D$ Bravais lattices: 1 oblique, 2 rectangular, 3 centered rectangular (rhombic), 4 hexagonal, and 5 square. Black arrows represent lattice vectors in the Bravais lattices. Courtesy of Wikipedia.}
  \label{fig:Bravais}
\end{figure}

Existing methods however often fall short for crystal image analysis
with a variety of crystal patterns, fine features, and complicated geometry
(see Figure \ref{fig:ex0} for examples). For instance, when the
out-of-focus problem occurs, usually in optical microscopies and
bright-field microscopies, the image intensity at the centers of
atoms might vary a lot over the imaging domain, making it difficult to
determine the reference configuration for atoms (see Figure
\ref{fig:it} (left) for an example).

In this paper, we propose a two step procedure for crystal image
analysis in these challenging scenarios. In the first step, an
efficient phase-space sketching method is used to classify complicated
crystal configurations and determine reference crystal patterns. In
the second step, once the reference crystal patterns are learned from
the first step, a recently developed crystal image analysis method
based on two-dimensional synchrosqueezed transforms
\cite{YangLuYing:2015,LuWirthYang:2016} is applied to identify
dislocations, cracks, grain boundaries, crystal orientation,
deformation, and possibly other useful information.

\begin{figure}
  \begin{center}
  \begin{tabular}{cc}
      \raisebox{0.1in}{    \includegraphics[height=1in]{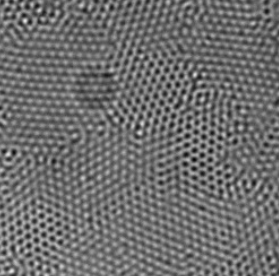}}& \includegraphics[height=1.3in]{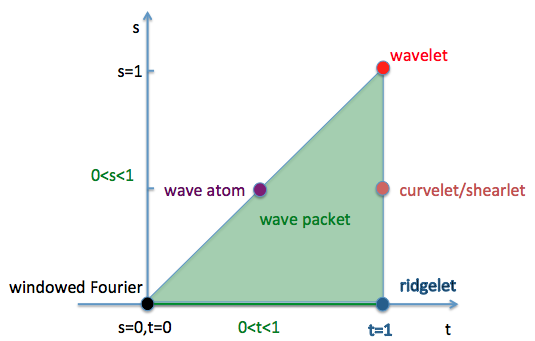}
     \end{tabular}
  \end{center}
  \caption{Left: An example of crystal images with the out-of-focus problem,
    in which an atom might be a black or a white dot in the crystal
    image. It is difficult to determine the position of atoms in this
    case. Courtesy of Yilong Han of Hong Kong University of Science
    and Technology. Right: The relation of windowed Fourier transform, wavelet transform, curvelet/shearlet transform, ridgelet transform, wave atom transform, and the wave packet transform in \cite{YangYing:2013,YangYing:2014}.}
  \label{fig:it}
\end{figure}

A major difficulty of crystal pattern classification comes from the
considerable variability within object classes and the inability of
existing distances to measure image similarities. Part of this
variability is due to grain orientation, crystal deformation, and
defects; another part results from the imaging variation, e.g., the
difference of image illumination, light reflection, and out-of-focus
issues. An ideal transformation for classification should provide a
representation invariant to such uninformative variability. Note that
a representation that is completely invariant with respect to
deformation would not distinguish different Bravais lattices (see
Figure \ref{fig:Bravais}), since the lattices are equivalent up to
affine transforms. The crystal image representation must therefore be
capable of distinguishing different lattices, while being invariant to
small elastic deformation due to external forces on grains.

Texture classification has been extensively studied in the
literature. Translation and rotation invariant representations have
been standard tools: these representations can be constructed with
autoregression methods \cite{4767811}, hidden Markov models
\cite{536891}, local binary patterns \cite{Guo2010706}, Gabor or
wavelet transform modulus with rotation invariant features
\cite{743859}. Scale and affine invariance has been also studied
recently using affine adaption
\cite{Lazebnik:2005:STR:1070616.1070813}, fractal analysis
\cite{Xu:2006:PIT:1153171.1153679}, advanced learning
\cite{Varma:2002:CIM:645317.649321}, and combination of filters
\cite{Zhang2003657}. More recently, deep convolution networks
\cite{CPA:CPA21413} together with advanced learning techniques
\cite{6522407} have been applied to design deformation invariant
representations \cite{Sifre:2013}. Combined with advanced learning techniques, the convolution neural network (CNN) and the scattering transform can provide deformation invariance and could be used for the problem discussed in this paper. To achieve transformation invariance using deep learning, most current methods usually make use of dataset augmentation; but this requires larger number of model parameters and training data, because the learned model needs to capture enough features for all the possible transformations of the input. Hence, methods in this direction is very expensive. A recent work \cite{Pooling} combines the pooling idea and deep CNN to achieve transformation invariance while keeping the computational expense relatively low. However, all crystal lattices are equivalent up to an affine transformation (see Figure 2 for examples) and we only need the invariance to crystal elastic deformation that is relatively much smaller than the affine transformation. How to tune the extend of deformation invariance via the method in \cite{Pooling} to match the demand in crystal image analysis is still worth to explore. On the other hand, since our problem is more specific than general image classification, we aim to design a more efficient and specific method, taking advantage of the periodic structure of crystal images. The proposed phase-space sketching is a nonlinear operator that rescales, shifts, and coarsens (i.e., pooling) the synchrosqueezed (SS) energy distribution. Since the SS energy distribution has already been computed for the analysis of crystal orientation, defects, and deformation, the computational cost of which scales nearly linearly in terms of the image size (see \cite{YangLuYing:2015}), the additional computational cost for the phase-space sketch is almost negligible, since we just need to sum up the synchrosqueezed representation using a coarser grid, which scales linearly in the size of the phase-space representation by the synchrosqueezed transform. 

The phase-space sketching proposed in this paper is an alternative
method for the invariant texture classification based on the
two-dimensional synchrosqueezed transform
\cite{YangYing:2013,YangYing:2014}. The phase-space sketching encodes
microscopic crystal images in a translation, rotation, illumination,
and scale invariant representation. This new representation is stable
to deformation and invariant to a class of elastic deformation in
materials. The extent of the elastic deformation invariance is
specified by user-defined parameters. As we shall see later, the
results of the two-dimensional synchrosqueezed transforms can be
applied to identify dislocations, cracks, grain boundaries, crystal
orientation and deformation following the algorithms in
\cite{YangLuYing:2015,LuWirthYang:2016}. The proposed phase-space
sketching only adds minimal computational cost to our existing
methodologies for atomic resolution crystal image analysis.  We
  remark that it is also possible to extend existing approaches of
  texture segmentation methods combined with classification procedure
  to achieve our goal of crystal image analysis of complex geometry in
  this paper, the advantage of the proposed method is to put all these
  image analysis steps in a single framework based on synchrosqueezed
  transforms, which have been shown to be suitable for crystal image
  analysis.

The rest of this paper is organized as follows. We start by
introducing the mathematical model of atomic resolution crystal images
and our previous analysis framework on crystal image analysis in
Section~\ref{sec:model}. In Section~\ref{sec:sk}, we propose the phase
space sketching based on the two-dimensional synchrosqueezed transform
to identify reference crystal patterns. With the reference crystal
patterns ready, a complete algorithm for invariant texture
classification and segmentation is introduced. We apply the whole
algorithm to several real examples in materials science in
Section~\ref{sec:analysis}. Finally, we present a brief summary of the
proposed methodology in Section~\ref{sec:conclusion}.

\section{Crystal image modeling and synchrosqueezed transform (SST)}
\label{sec:model}

\subsection{Mathematical models for atomic resolution crystal image}

Consider an atomic resolution $2D$ image (experimentally, this is
often done for a thin slice of a $3D$ polycrystalline material) that
may consist of multiple grains with the same Braivas lattice.  Denote
the reference Bravais lattice as
\begin{equation*}
\mathcal L=\{av_1+bv_2\,:\,a,b\text{ integers}\}\,,
\end{equation*}
where $v_1,v_2\in\RR^2$ represent two fixed linearly independent lattice vectors (see Figure \ref{fig:Bravais} for examples).
Let $s(2\pi x)$ be the shape function describing a single perfect unit
cell in the image, extended periodically in $x$ with respect to the
reference crystal lattice.  We denote by an open set $\Omega_k$,
$k = 1, \ldots, M$, the grains the system consists of, and by
$\Omega=\cup \Omega_k$ the domain occupied by all grains. We only consider disjoint sets $\{\Omega_k\}_k$ in this paper. In more complicated cases when there are overlapping grains, it would be interesting to extend the mode decomposition techniques for one-dimensional oscillatory signals in \cite{RDBR,MMD} to two-dimensional so as to decompose overlapping grains. Hence, not considering grain boundaries, 
the polycrystal image $f:\Omega\to\RR$ can be modeled as
\begin{equation}
f(x)= \sum_{k=1}^M  \chi_{\Omega_k} (x) (\alpha_k(x) s(2\pi N \phi_k(x))+c_k(x)),
\label{eqn:crystal}
\end{equation}
where $N$ is the reciprocal lattice parameter (or rather the relative
reciprocal lattice parameter since the dimension of the image is
normalized in our analysis). The crystal image model above works for
both simple lattices (Braivais lattices listed in Figure
\ref{fig:Bravais}) and complex lattices (such that the unit cell
consists more than one atoms); extensions to more complicated images
will be considered below. Here $\chi_{\Omega_k}$ is the indicator
function of each grain $\Omega_k$; $\phi_k:\Omega_k \to \RR^2$ maps
the atoms of grain $\Omega_k$ back to the configuration of a perfect
crystal, i.e., it can be understood as the inverse of the elastic
deformation. The local inverse deformation gradient is then given by
$ \nabla \phi_k$ in each $\Omega_k$.  The smooth amplitude envelope
$\alpha_k(x)$ and the smooth trend function $c_k(x)$ in
\eqref{eqn:crystal} model possible variation of intensity,
illumination, etc.~during the imaging process.  By the $2D$ Fourier
series $\widehat s$ of $s$, we can rewrite \eqref{eqn:crystal} as
\begin{equation}
\begin{aligned}
  f(x) =\sum_{k=1}^M \chi_{\Omega_k} (x)\biggl( \sum_{\xi \in
      \mc{L}^{\ast}}\alpha_k(x) \widehat{s}(\xi)e^{2\pi  iN \xi \cdot
      \phi_k(x)}+c_k(x)\biggr),
\label{eqn:imagefunction}
\end{aligned}
\end{equation}
where $\mc{L}^{\ast}$ is the reciprocal lattice of $\mc{L}$ (recall
that the shape function $s$ is periodic with respect to the lattice
$\mc{L}$).

\begin{figure}[ht!]
  \begin{center}
    \begin{tabular}{cc}
     \includegraphics[height=2in]{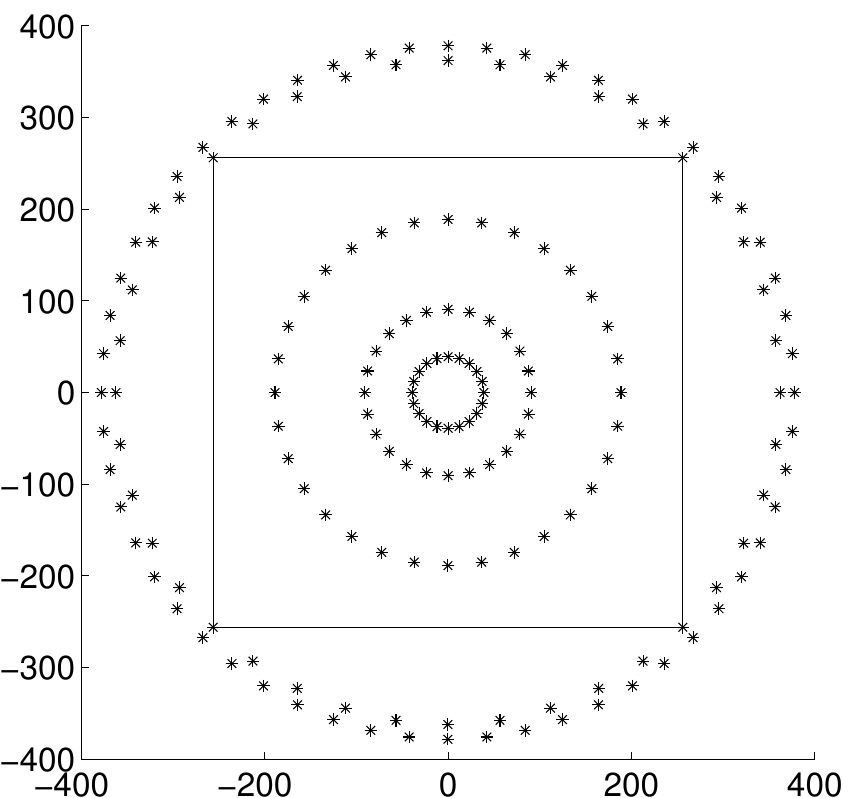} & \includegraphics[height=2in]{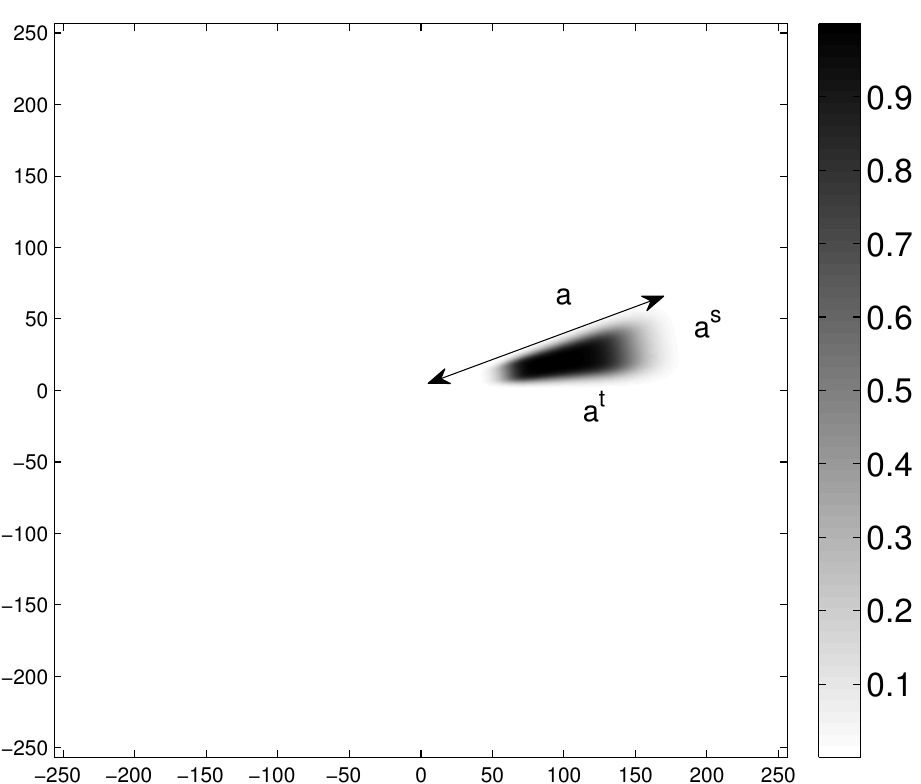}\\
    \end{tabular}
  \end{center}
  \caption{Left: Partition of unity in the Fourier domain for an image of size $512\times512$. Each point represents the center of the support of a bump function. The bump function centered at the origin is supported on a disk and is not indicated in this picture. Right: An example of a fan-shaped bump function centered at a point with radius $a$. The shape of bump functions is controled by two scaling parameters $(t,s)$.}
  \label{fig:tiling}
\end{figure}

\begin{figure}[ht!]
  \begin{center}
    \begin{tabular}{ccc}
      \includegraphics[height=1in]{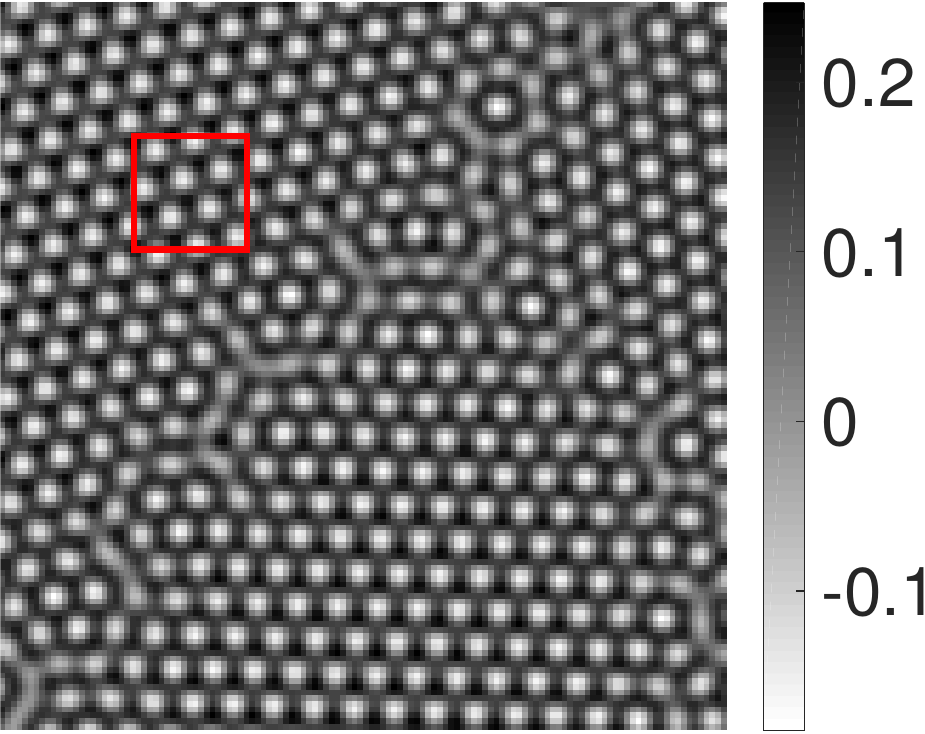} &  \includegraphics[height=1in]{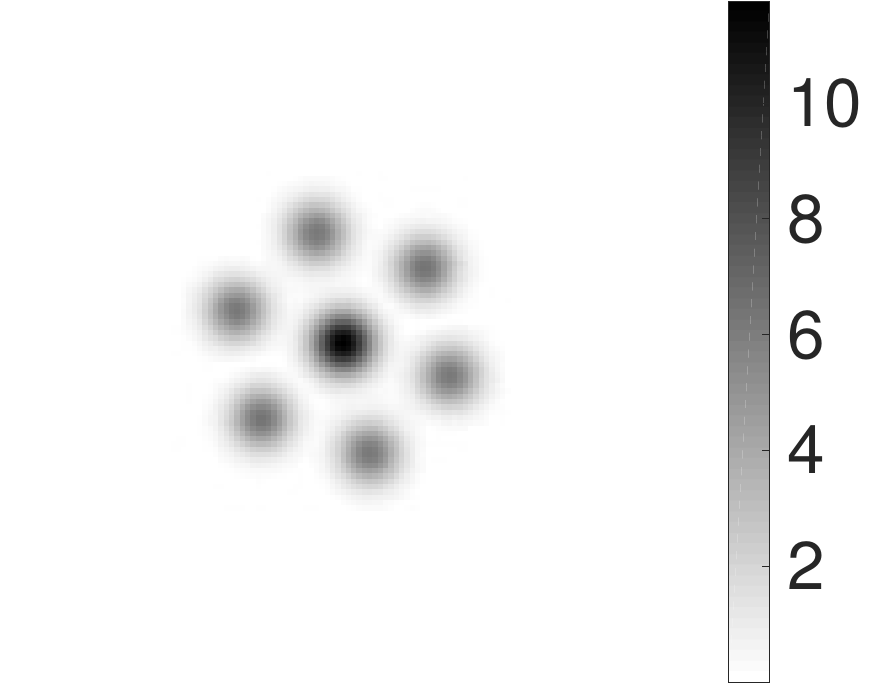}  &   \raisebox{-0.1in}{     \includegraphics[height=1.15in]{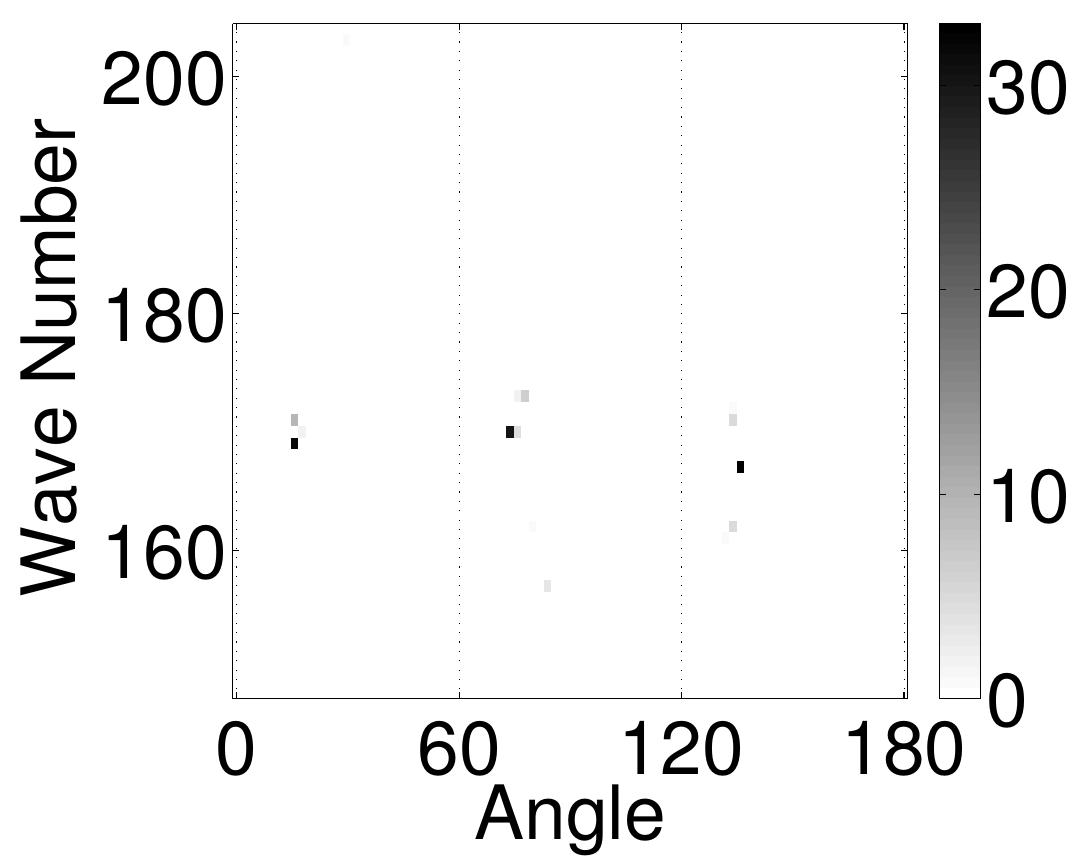} }  \\
       (a)& (b) &  (c) \\ 
   \raisebox{-0.1in}{ \includegraphics[height=1.15in]{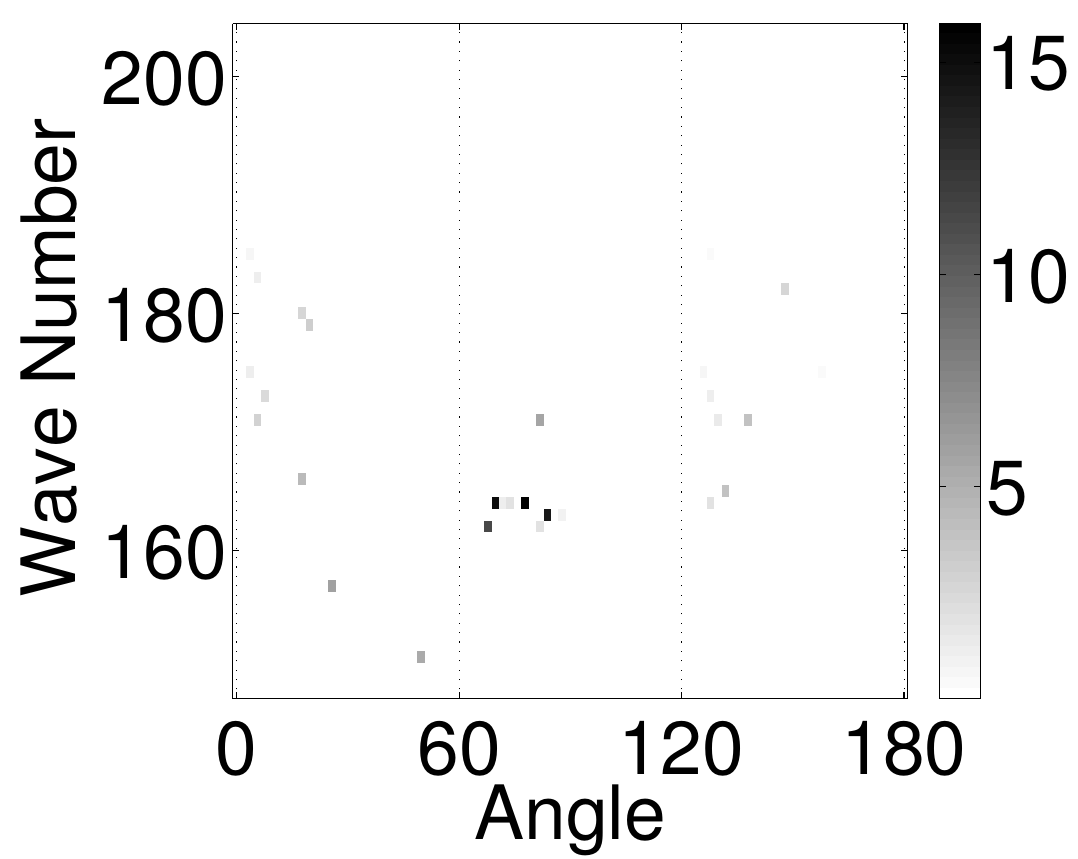} } &  \includegraphics[height=1in]{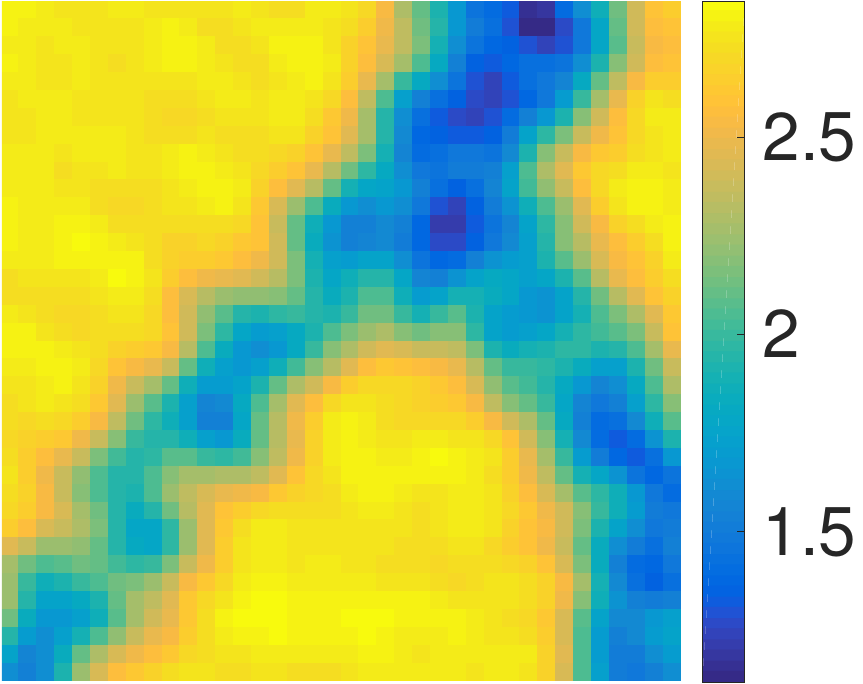}&  \includegraphics[height=1in]{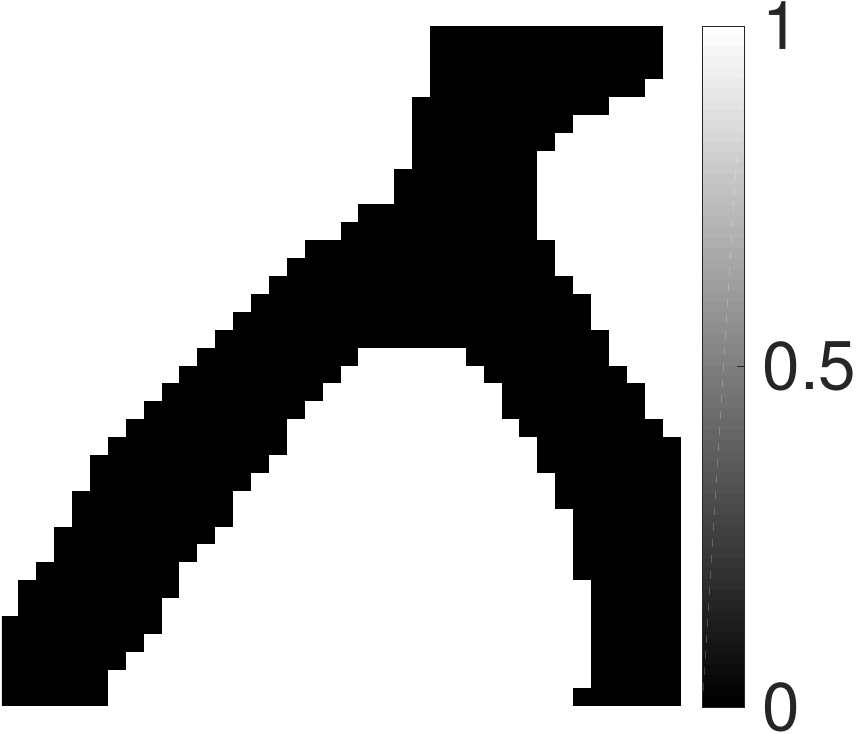}\\
    (d) &   (e)&(f)
    \end{tabular}
  \end{center}
  \caption{(a) An example of a crystal image. (b) Windowed Fourier
    transform at a local patch indicated by a rectangle. (c) The SS energy
    distribution in polar coordinates at a point outside the defect
    region. (d) The SS energy distribution at a point in the defect region.  (e) The defect
    indicator $\mass(x)$. (f) Identified defect region by thresholding $\mass(x)$. }
  \label{fig:patch}
\end{figure}

\begin{figure}[ht!]
  \begin{center}
    \begin{tabular}{cccc}
      \includegraphics[height=1in]{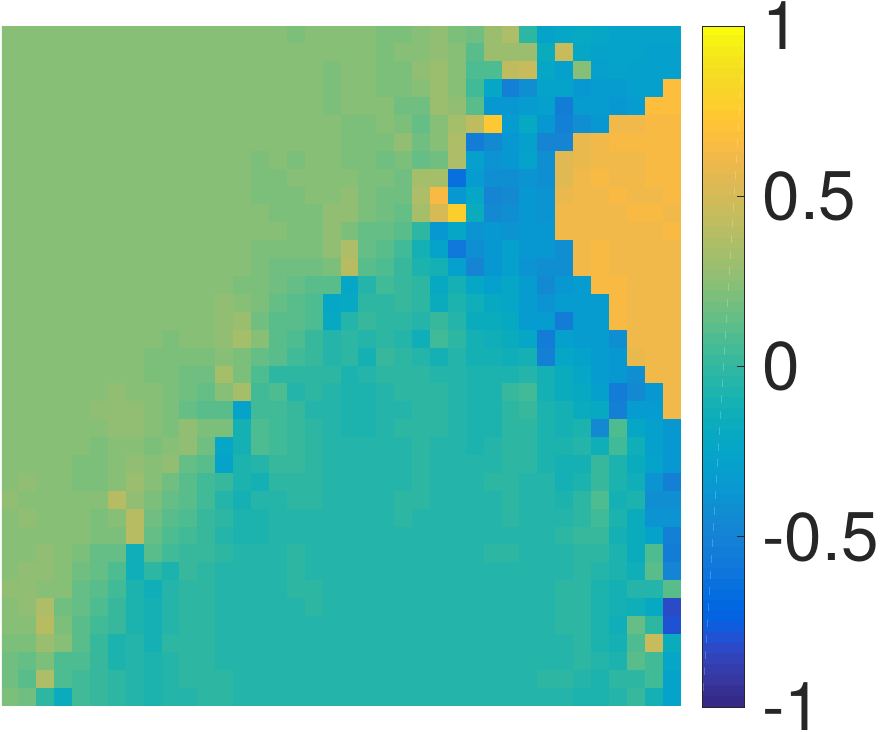} &  \includegraphics[height=1in]{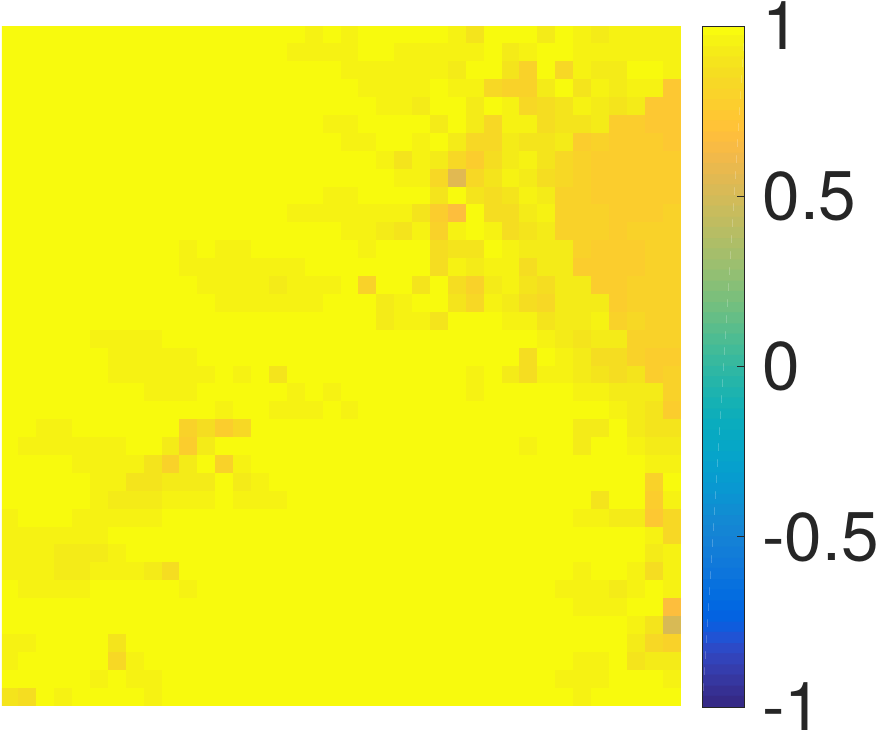} &   \includegraphics[height=1in]{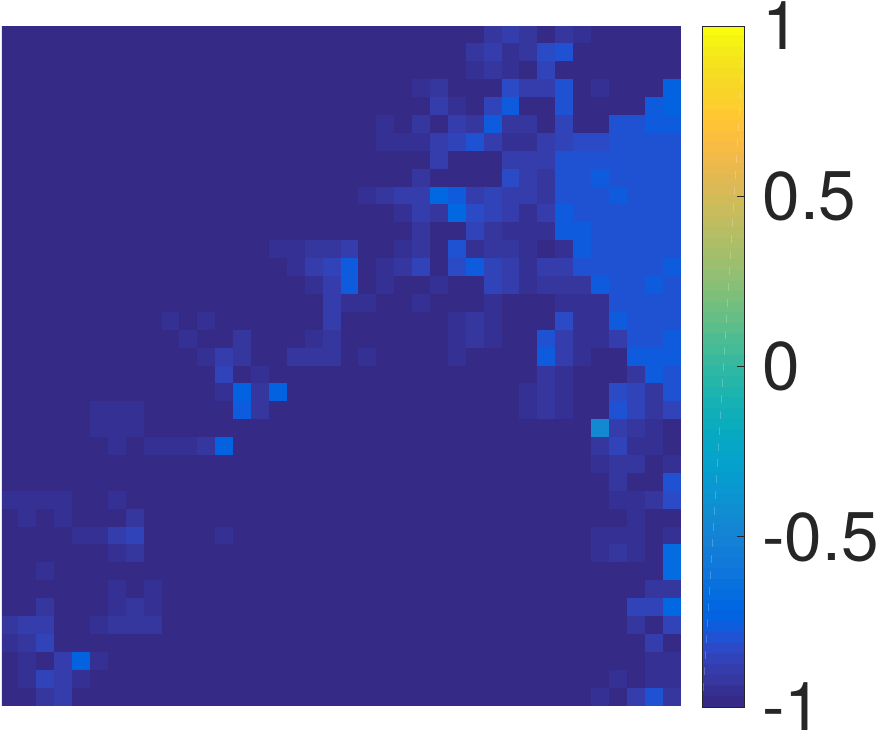} &  \includegraphics[height=1in]{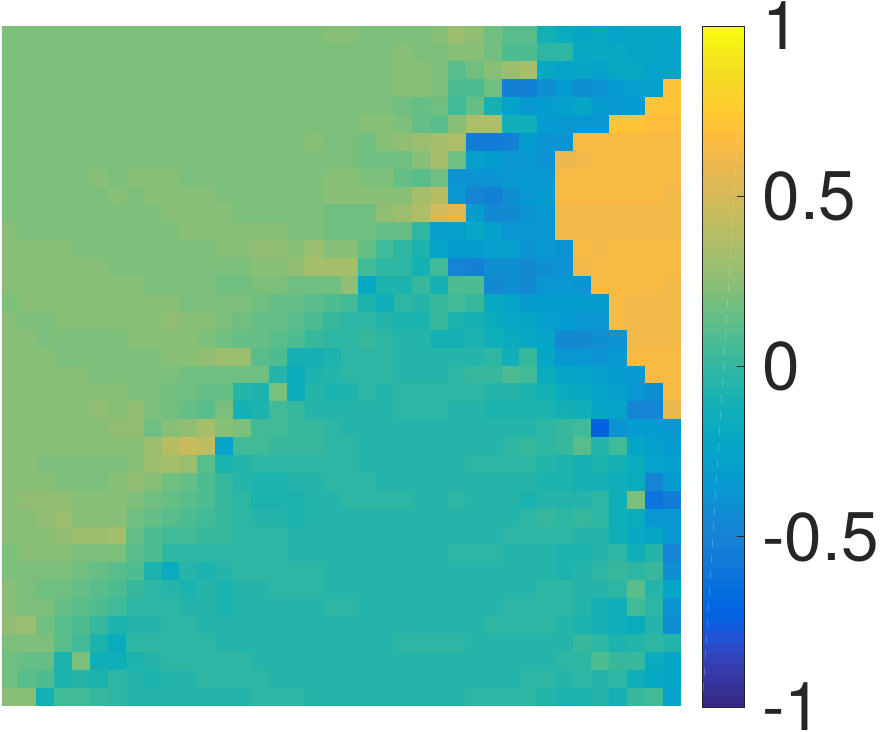}\\
      $G_0^{11}$& $G_0^{21}$ &   $G_0^{12}$ & $G_0^{22}$\\
    \end{tabular}
    \\
        \begin{tabular}{ccc}
      \includegraphics[height=1in]{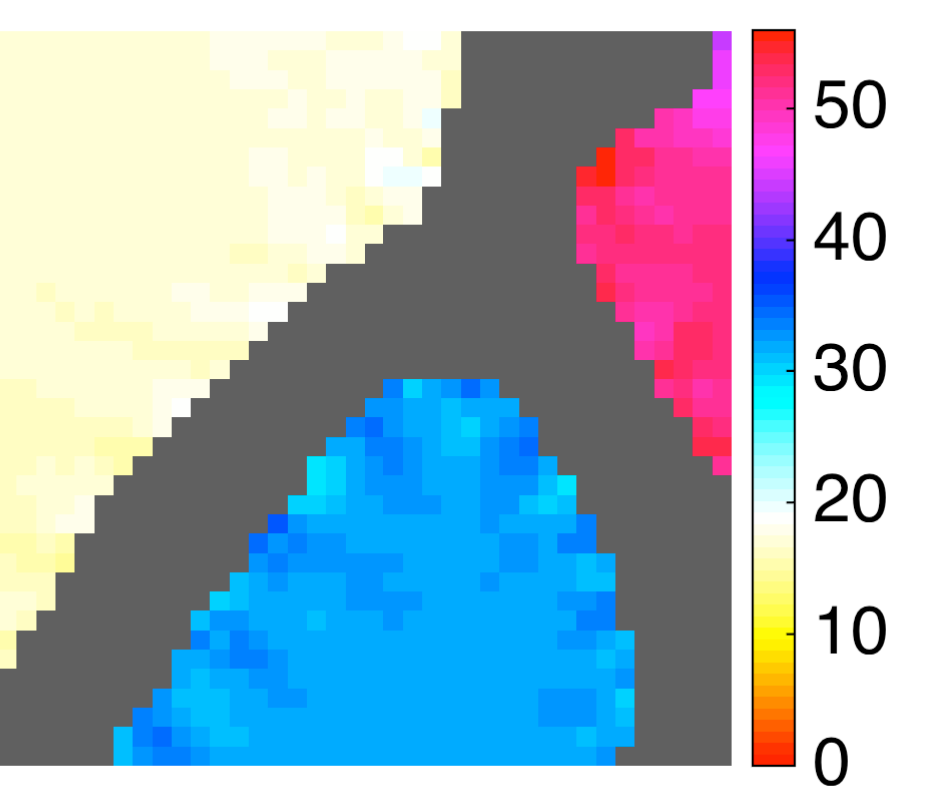} &  \includegraphics[height=1in]{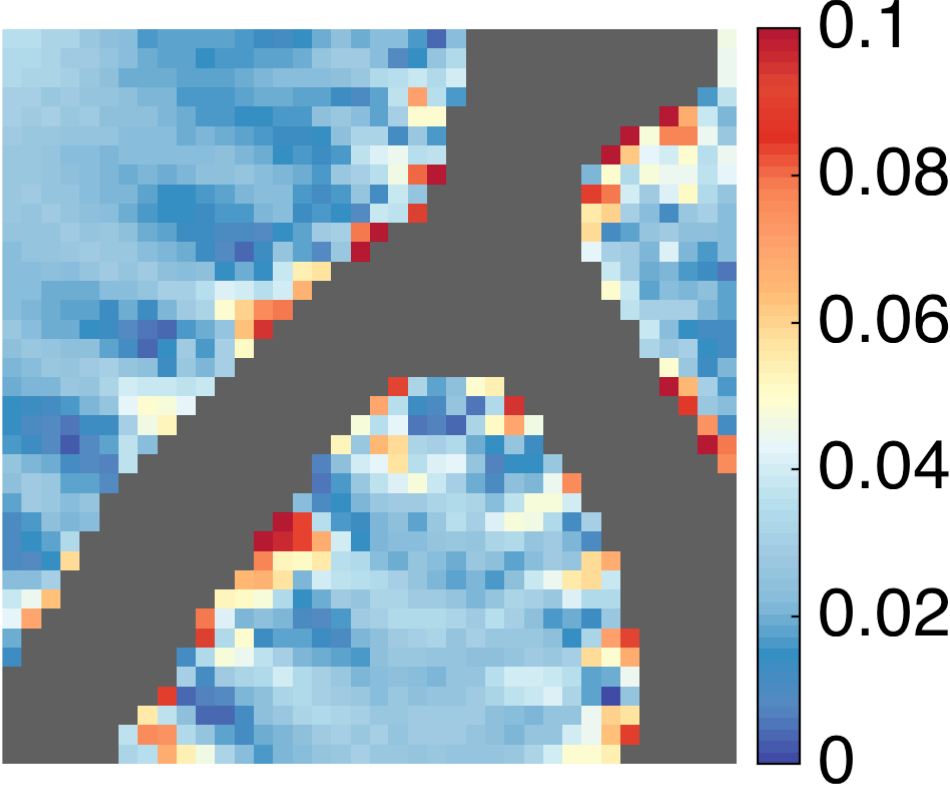} &      \includegraphics[height=1in]{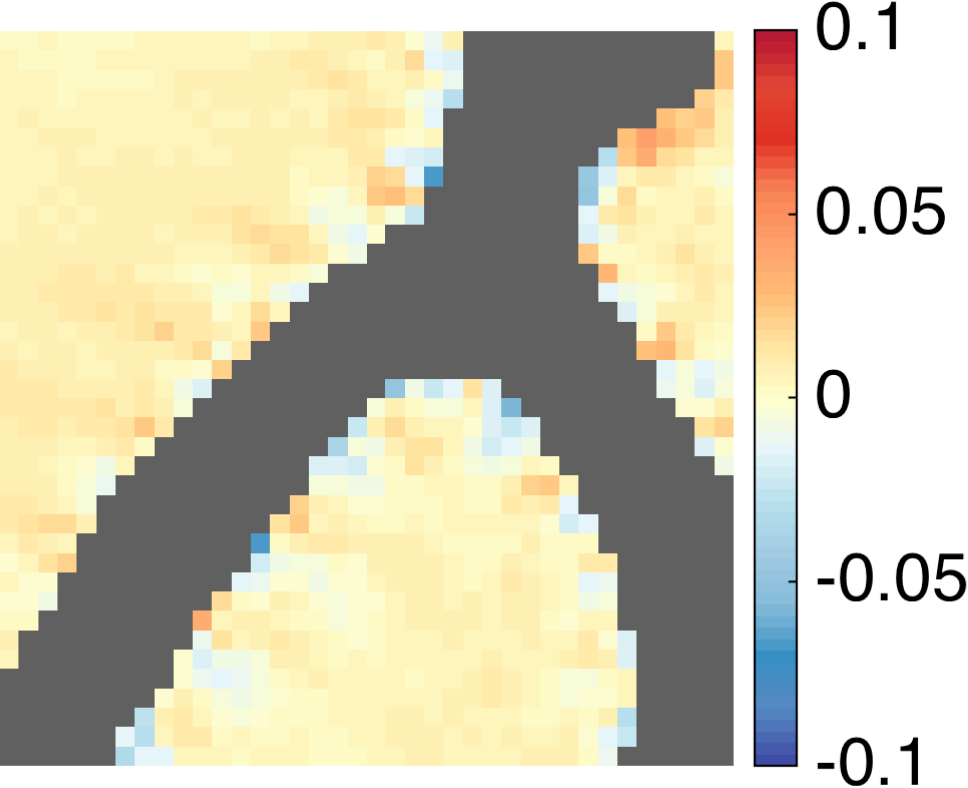}
    \end{tabular}
  \end{center}
  \caption{Top panel: Estimated inverse deformation gradient $G_0\in \RR^{2\times 2}$ of the atomic
  resolution crystal image in Figure \ref{fig:patch} (a). Bottom panel: The crystal orientation, the difference in principal stretches, and the volume distortion of $G_0$. The grey mask in these figures is the defect region identified in Figure~\ref{fig:patch} (f).}
  \label{fig:patchG}
\end{figure}

\subsection{Synchrosqueezed transform (SST)}

As shown in \cite{YangLuYing:2015,LuWirthYang:2016}, the 2D SST is an
efficient tool to estimate the defect region and also the local
inverse deformation gradient 
$G_0(x)=\sum_{k}\chi_{\Omega_k}(x) N \nabla \phi_k(x)$ in the interior of each
grain $\Omega_k$. The main observation is that, in each grain
$\Omega_k$, the image is a superposition of planewave\footnote{Planewave is a wave with constant frequency or wave vector.}-like components
$\alpha_k(x) \widehat{s}(\xi) e^{2\pi i N \xi \cdot \phi_k(x)}$ with local
wave vectors $N \nabla(\xi \cdot \phi_k(x))$. It has been shown in
\cite{YangYing:2013,YangYing:2014} that the SST is able to estimate
the local wave vectors accurately. Based on the local wave vector
estimation, we can compute the inverse deformation gradient $G_0$ via
a least square method, see \S\ref{sec:SSTCA} below.

The starting point of $2$D SST is a wave packet $w_{a\theta x}$,
which is, roughly speaking, constructed by translating, rotating, and rescaling or modulating a mother wave packet $w:\RR^2\to\CC$ according to the spatial center parameter $x\in\RR^2$, the angular parameter $\theta\in[0,2\pi)$ in the frequency domain, and the radial parameter $a\in\RR$
in the frequency domain \cite{YangYing:2013,YangYing:2014,YangLuYing:2015,LuWirthYang:2016}. To introduce the definition of wave packets, let us define the following notations:
\begin{enumerate}
\item The scaling matrix
\[A_a=\left( \begin{array}{cc}
a^t&0\\
0&a^s 
\end{array}\right),\]
where $a$ is the distance from the center of one wave packet to the origin in the Fourier domain.
\item A unit vector $e_\theta=(\cos \theta,\sin \theta)^T$ with a rotation angle $\theta$.
\item $\theta_\alpha$ represents the argument of a given vector $\alpha$.
\item $w(x)$ of $x\in \RR^2$ denotes a mother wave packet, which is in the Schwartz class and has a non-negative, radial, real-valued, smooth Fourier transform
$\widehat{w}(\xi)$ with a support equal to a ball $B_d(0)$ centered at the origin with a radius $d\leq 1$ in the Fourier domain. The mother wave packet is required to obey the admissibility condition: $\exists 0<c_1(t,s)<c_2(t,s)<\infty$ such that
\[
c_1(t,s)\leq \int_0^{2\pi} \int_1^\infty a^{-(t+s)} |\widehat{w}(A^{-1}_a R^{-1}_\theta (\xi-a\cdot e_\theta))|^2 a\ud a\ud\theta \leq c_2(t,s)
\]
for any $|\xi|\geq 1$.
\end{enumerate}
\begin{defn}
  \label{def:GC2D}
For $\frac{1}{2}<s\leq t<1$, define $\widehat{w_{a\theta b}}(\xi)=\widehat{w}(A^{-1}_a R^{-1}_\theta (\xi-a\cdot e_\theta))e^{-2\pi ib\cdot\xi}a^{-\frac{t+s}{2}}$ as a wave packet in the Fourier domain. Equivalently, in the space domain, the corresponding wave packet is
\begin{align*}
w_{a\theta b}(x)=a^{\frac{t+s}{2}}e^{2\pi ia(x-b)\cdot e_\theta}w(A_a R^{-1}_\theta(x-b)).
\end{align*}
In such a way, a family of wave packets $\{w_{a\theta b}(x),a\in [1,\infty),\theta\in[0,2\pi),b\in \RR^2\}$ is constructed. 
\end{defn}

\begin{defn}
  \label{def:GCT}
The wave packet transform\footnote{Throughout this paper all the numerical implementation of integration follows the rectangle rule. For the purpose of simplicity, we don't specify the discretization of all variables in this paper. In a general setting, we assume the given image of size $N_0\times N_0$ is defined on $[0,1]^2$ with grid points $\{(\frac{i}{N_0},\frac{j}{N_0}) \}$ for $0\leq i, j\leq N_0-1$; the frequency domain in a Cartesian coordinate is discretized with grid points $\{(i,j) \}$ for $-\frac{N_0}{2}\leq i, j\leq \frac{N_0}{2}-1$; unless specified, the frequency domain in a polar coordinate is restricted in the range $[0,\frac{\sqrt{2}N_0}{2}]$ or $[1,\frac{\sqrt{2}N_0}{2}]$  for radius and $[0,\pi]$ for angle, and they are discretized uniformly with step size $1$ in radius and $\pi/90$ in angle. The code for phase-space sketching and the test data are available online in the SynCrystal package for reference (https://github.com/SynCrystal/SynCrystal). }
 of a function $f(x)$ is a function 
\begin{align*}
  W_f(a,\theta,b)&= \int_{\RR^2}\overline{w_{a\theta b}(x)}f(x)\ud x
\end{align*}
for $a\in [1,\infty)$, $\theta\in[0,2\pi)$, $b\in \RR^2$. For convenience, we also set $W_f(a,\theta,b) = 0$ for $a\in[0,1)$.
\end{defn}

More specifically, the discrete wave packets are constructed via the partition of unity in the frequency domain (see Figure \ref{fig:tiling}) so that they have compact supports in the frequency domain. The Fourier transform of a wave packet is essentially the product of a bump function and a plane wave. Figure \ref{fig:tiling} (left) illustrates the positions of the centers of bump functions and Figure \ref{fig:tiling} (right) shows an example of a bump function when its center has a radius $O(a)$. The reader is refered to Section $3$ of \cite{YangYing:2014} for detailed implementation of the discret wave packet transform. As shown in Figure \ref{fig:tiling} (right), there are two scaling parameters $(t,s)$ to control the geometry of the supports of wave packets. If $(t,s)=(1,1)$, the wave packet transform is essentially the wavelet transform; if $(t,s)=(1,1/2)$, the wave packet transform becomes the curvelet transform (see Figure \ref{fig:it} (right) for an illustration). The wave packet transform is a generalization of curvelet and
wavelet transforms with better flexibility in frequency scaling and
consequently is better suited to analyze crystal images with complicated 
geometry. As a convolution with smooth wave packets, $W_f$ is well-defined and smooth even under very low regularity requirements for $f$, e.\,g.\ $f\in L^\infty(\RR^2)$. 

In contrast to the windowed Fourier transform of a given crystal
image, whose spectrum spreads out in the phase space as illustrated in
Figure~\ref{fig:patch}(b), the $2$D synchrosqueezed transform (SST)
based on wave packets aims at a sharpened phase-space
representation. In the SST, for each $(a, \theta, x)$, we define the
corresponding local wave vector estimation
\begin{equation}\label{eqn:lwv}
  v_f(a, \theta,x) = \Re \frac{\nabla_x
    W_f(a,\theta,x)}{2\pi i W_f(a,\theta,x)}
\end{equation}
for $|W_f(a,\theta,x)|>\gamma$, where $\gamma$ is a threshold
parameter to ensure stability of the estimation. Here, $\nabla_x W_f$
denotes the gradient of $W_f$ with respect to its third argument $x$.

If $f(x)$ is a plane wave with a wave vector
$(v_1,v_2)\in \mathbb{R}^2$, then simple
algebraic calculation shows that the local wave vector estimation
$v_f(a,\theta,x)$ is exactly $(v_1,v_2)$ whenever $W_f(a,\theta,x)$ is
not zero. When $f(x)$ is a superposition of planewaves with
well-separated local wave vectors $\{(v^{(k)}_1,v^{(k)}_2)\}_k$, by
the stationary phase approximation,
$v_f(a,\theta,x)\approx (v^{(k)}_1,v^{(k)}_2)$ when
$(a\cos\theta,a\sin\theta)$ is close to some local wave vector
$(v^{(k)}_1,v^{(k)}_2)$. In the case of crystal images, locally $f(x)$
is a superposition of deformed planewaves. By applying Taylor
expansion to make approximations, one can also show that the local
wave vector estimation $v_f(a,\theta,x)$ can still approximate the
local wave vectors of $f(x)$. Even in the presence of heavy noise, this
approximation is still reasonably good by applying the SST based on
highly redundant wave packet frames \cite{Robustness,Canvas}. More precisely, let us revisit the theory developed in \cite{YangLuYing:2015} to support the argument just above.

\begin{defn}[$2D$ general shape function]
\label{def:GSF}

The $2D$ general shape function class ${S}_M$ consists of periodic functions $s(x)$ with a periodicity $(2\pi,2\pi)$, a unit $L^2([-\pi,\pi]^2)$-norm, and an $L^\infty$-norm bounded by $M$ satisfying the following conditions:
\begin{enumerate}
\item The $2D$ Fourier series of $s(x)$ is uniformly convergent;
\item $\sum_{n \in \ZZ^2} |\widehat{s}(n)|\leq M$ and $\widehat{s}(0,0)=0$;
\item Let $\Lambda$ be the set of integers $\{|n_1|\in \NN: \widehat{s}(n_1,n_2)\neq 0\text{ or }\widehat{s}(n_2,n_1)\neq 0\text{ for some }n_2\in\ZZ\}$. 
The greatest common divisor of all the elements in $\Lambda$ is $1$.
\end{enumerate}
\end{defn}

\begin{defn}[$2D$ general intrinsic mode type function (GIMT)]
  \label{def:GIMTF}
  A function $f(x)=\alpha(x)s(2\pi N  \phi(x))$ is a $2D$ GIMT of type $(M,N)$, if $s(x)\in {S}_M$, $\alpha(x)$ and $\phi(x)$ satisfy the conditions below.
\begin{align*} 
&\alpha(x)\in C^\infty, \quad |\grad\alpha|\leq M, \quad 1/M \leq \alpha\leq M, \\
 &\phi(x)\in C^\infty, \quad 1/M \leq \left|  \nabla (n^{\TT}  \phi) / \abs{n^{\TT} } \right| \leq M,\quad \text{and}\\
  &\left| \nabla^2 (n^{\TT}  \phi) / \abs{n^{\TT}} \right|\leq M, \quad \forall n \in \ZZ^2 \quad \text{s.t.}\quad \widehat{s}(n)\neq 0.
   \end{align*}
\end{defn}

The theorem below shows that the local wave vector estimation precisely estimates the local wave vectors of the crystal image at the points away from boundaries (see \cite{YangLuYing:2015} for the proof).

\begin{theorem}
  \label{thm:main}
  For a $2D$ GIMT $f(x)$ of type $(M,N)$ with any $\eps>0$ and any $r>1$, we let $\gamma = \sqrt{\epsilon}$ and define
  \begin{equation*}
  R_{\eps} = \left\{(a,\theta,b): |W_f(a,\theta,b)|\geq \gamma,\quad a\leq 2MNr \right\}
  \end{equation*}
  and 
  \[
  Z_{n} = \left\{(a,\theta,b):\left|A^{-1}_a R^{-1}_\theta\left(a\cdot e_\theta-N \nabla(n^{\TT}   \phi(b)) \right)\right|\leq d,\quad a\leq 2MNr \right\}
  \]
For fixed $M$, $r$, $s$, $t$, $d$, and $\epsilon$, there exists $N_0(M,r,s,t,d,\eps)>0$ such that for any $N>N_0(M,r,s,t,d,\eps)$ and a $2D$ GIMT $f(x)$ of type $(M,N)$, the following statements
  hold.
  \begin{enumerate}
  \item $\left\{Z_n: \widehat{S}(n)\neq 0 \right\}$ are disjoint and $R_{\eps}
    \subset \bigcup_{\widehat{S}(n)\neq 0} Z_{n}$;
  \item For any $(a,\theta,b) \in R_{\eps} \cap Z_{n}$, 
    \[
    \frac{\left|v_f(a,\theta,b)-N \nabla(n^{\TT}   \phi(b))  \right|}{ \left|N \nabla(n^{\TT}  \phi(b))  \right|}\lesssim\sqrt \epsilon.
    \]
  \end{enumerate}
\end{theorem}

For simplicity, the notations $O(\cdot)$, $\lesssim$ and $\gtrsim$ are used when the implicit constants may only depend on $M$, $s$, $t$, and $d$.

Motivated by the property of the local wave vector estimation $v_f$, the synchrosqueezed (SS) energy distribution of $f$ is constructed as
\begin{equation}
T_{\!f}(v,x) \!=\! \iint_{(a,\theta)\in D(x,\gamma)} \!\!|W_{\!f}(a,\theta,x)|^2 \delta(v_{\!f}(a,\theta,x)-v) \, a\ud a\ud\theta\,, \label{eq:SED}
\end{equation}
where
\[D(x,\gamma):=\left\{ (a,\theta) \in (0, \infty) \times (0,2\pi):    |W_f(a,\theta,x)|>\gamma\right\}\]
and $\delta$ in \eqref{eq:SED} denotes the Dirac measure. In the numerical implementation of $T_{\!f}(v,x) $ in \eqref{eq:SED}, a normalized Gaussian function with a sufficiently small support (of the same order as the discretization step size of its domain) is applied to approximate the Dirac delta function. The discretization of $a$ and $\theta$ is given by the partition of unity of the discrete wave packet transform. The variable $x$ is discretized uniformly in the Cartesian coordinate, while the variable $v$ is discretized in the polar coordinate.  The SST
squeezes the wave packet spectrum $|W_f(a, \theta, x)|^2$ according to
$v_f(a, \theta, x)$ to obtain a sharpened and concentrated
representation of the image in the phase space. Hence, in the interior
of a grain, the SS energy distribution $T_f$ has a support
concentrating around local wave vectors
$N\nabla (\xi\cdot \phi_k(x))$, $ \xi \in \mc{L}^{\ast}$, and is given
approximately by (see e.g., Figure~\ref{fig:patch}(c) in polar
coordinate)
\begin{equation}
  T_f(v, x) \approx \sum_{\xi \in \mc{L}^{\ast}} \alpha_k(x)^2|\widehat{s}(\xi)|^2 
  \delta\bigl( v - N \nabla ( \xi\cdot \phi_k(x))\bigr), \label{eq:Tfapprox}
\end{equation}
understood in the distributional sense. Therefore, by locating the
energy peaks of $T_f$, we can obtain estimates of local wave vectors
$N\nabla(\xi\cdot \phi_k(x))$ and also their associated spectral
energy. In practice, we choose high energy peaks corresponding to
$\xi$ close to the origin in the reciprocal lattice to estimate the
inverse deformation gradient $G_0$ and grain boundaries to guarantee
numerical stability.

\subsection{SST based crystal image analysis}  \label{sec:SSTCA}
In the case of a crystal image that consists of only one kind of lattice in the image, it has been shown in \cite{YangLuYing:2015,LuWirthYang:2016} that the SST can be applied to estimate the inverse deformation gradient $G_0$, grain boundaries, and defects. For simplicity, let us focus on the case of hexagonal lattices in this section. It is straightforward to generalize to  other Bravais lattices. 
 
In the case of hexagonal lattices, we have six dominant reciprocal lattice vectors $\xi\in \mathcal{L}^*$ such that $|\widehat{s}(\xi)|$ is significant, which can be further reduced to three due to the
symmetry $\xi \leftrightarrow -\xi$. We will henceforth denote these
as $\xi_n$, $n = 1, 2, 3$, and denote by $v^{\text{est}}_n(x)$ the
estimate of $N \nabla (\xi_n \cdot \phi_k(x)) = N (\nabla \phi_k(x))
\xi_n$.  The inverse deformation gradient $G_0(x) =\nabla \phi_k(x)$ is
determined by a least squares fitting to identify a linear transformation that maps the reference reciprocal lattice vectors $N \xi_n$
to the estimated local wave vectors $v^{\text{est}}_n$:
\[
G_0(x) = \underset{G}{\argmin} \sum_{n=1}^3 |v^{\text{est}}_n(x) - N G \xi_n |^2.
\]
In practice, for each physical point $x$ we represent $T_f(\cdot, x)$
in polar coordinates $(r,\vartheta)\in[0,\infty)\times[0,\pi)$ (the
information in $\vartheta\in[\pi,2\pi)$ is redundant due to
symmetry). To identify the peak locations $\{v^{\text{est}}_n\}$, we
choose the grid point with highest amplitude in each $\pi/3$-degree
sector of $\vartheta$.

Since the local wave vector estimation is no longer valid around the crystal defects, the SS energy distribution $T_f$ does not have dominant energy peaks around local wave vectors. Hence, we may characterize the defect region by quantifying how concentrated the energy distribution is. One possible way is to use an indicator function as follows. For each $n \in \{1, 2, 3\}$ (corresponding to one of the sectors), we define
\[
w_n(x)=\frac{\displaystyle \int_{B_\delta(v^{\text{est}}_n)} T_f(v,x) \ud v}{\displaystyle \int_{\arg v \in[(n-1)\pi/3,n\pi/3)} T_f(v,b) \ud v}\,,
\]
where $B_{\delta}(v^{\text{est}}_n)$ denotes a small ball around the
estimated local wave vector $v^{\text{est}}_n$, and $\arg v$ means the argument of $v$. Hence, $\mass(x) :=
\sum_n w_n(x)$ will be close to $3$ in the interior of a grain due to
\eqref{eq:Tfapprox}, while its value will be much smaller than $3$
near the defects. This is illustrated in Figure~\ref{fig:patch}(e),
where we show $\mass(x)$ for the crystal image in
Figure~\ref{fig:patch}(a). The estimate of defect regions can be
obtained by a thresholding $\mass(x)$ at some value $\eta\in(0,3)$ according to
\begin{equation*}
\Omega_d=\{x\in\Omega\ :\ \mass(x)<\eta\}\,,
\end{equation*}
an illustration of which is shown in Figure~\ref{fig:patch}(f).
Figure~\ref{fig:patchG} shows the estimate of $G_0$ by the SST. When
showing the result, we have used a more transparent way to represent
the inverse deformation gradient $G_0$ via the polar decomposition
$G_0(x)=U_0(x)P_0(x)$ for each point $x\in \Omega$, where $U_0(x)$ is
a rotation matrix and $P_0(x)$ is a positive-semidefinite symmetric
matrix. At each position $x$, the crystal orientation can be estimated
via the rotation angle of $U_0(x)$; the volume distortion of $G_0(x)$
can be visualized by $\det(G_0(x))-1$; the quantity
$|\lambda_1(x)-\lambda_2(x)|$ characterizes the difference in the
principal stretches of $G_0(x)$ as a measure of shear strength, where
$\lambda_1(x)$ and $\lambda_2(x)$ are the eigenvalues of $P_0(x)$. The
bottom panel of Figure~\ref{fig:patchG} shows these quantities
corresponding to the estimate of $G_0$ in the top panel. In later
numerical examples, we will always present the estimated inverse
deformation gradient in the same fashion.

\section{Phase space sketching}
\label{sec:sk}

To extend the crystal image analysis to more complicated scenario, in
this section, we propose the phase-space sketching, which is a sparse
invariant representation of the phase-space information obtained by
synchrosqueezed transforms. The phase-space sketching will enable
classification of different crystal types presented in the same image
or across images.

\subsection{Invariant representation}
\label{sub:rep}

As compared to the crystal image model with a unique Bravais lattice
in the previous section, 
the mathematical model of a crystal
image that consists of multiple types of lattices can be written as in \eqref{eqn:crystal} and \eqref{eqn:imagefunction}. 

The image classification and segmentation problem is to
identify each domain $\Omega_k$ and classify its corresponding shape
function $s_k$. The main difficulty is due to the
considerable variability within object classes (e.g., grain
orientation, crystal deformation, defects, and difference of image
illumination). Our goal is to design a representation of the crystal
image that is invariant to most of these variability.

It has been shown in \cite{deep,6522407} that, deep
convolution networks have the ability to build large-scale invariants
stable to deformations. Combined with advanced learning techniques,
the convolution network like the scattering transform can provide
deformation invariance \cite{Sifre:2013} and could be used for the
problem discussed in this paper. On the other hand, since our problem
is more specific than general image classification, we aim to design a
more efficient and specific method. Taking advantage of the periodic
structure of crystal images, we introduce the phase-space sketching
instead of applying the convolution network for the purpose of
computational efficiency. The phase-space sketching is a nonlinear
operator rescaling, shifting, and coarsening the SS energy
distribution.

 \begin{figure}
  \begin{center}
  \begin{tabular}{ccc}
        \includegraphics[height=0.9in]{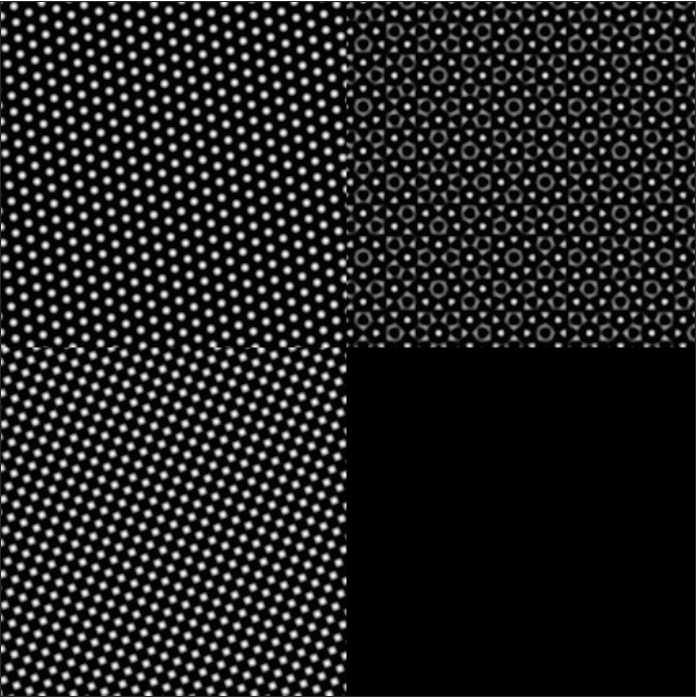}&  \raisebox{-0.05in} {\includegraphics[height=1in]{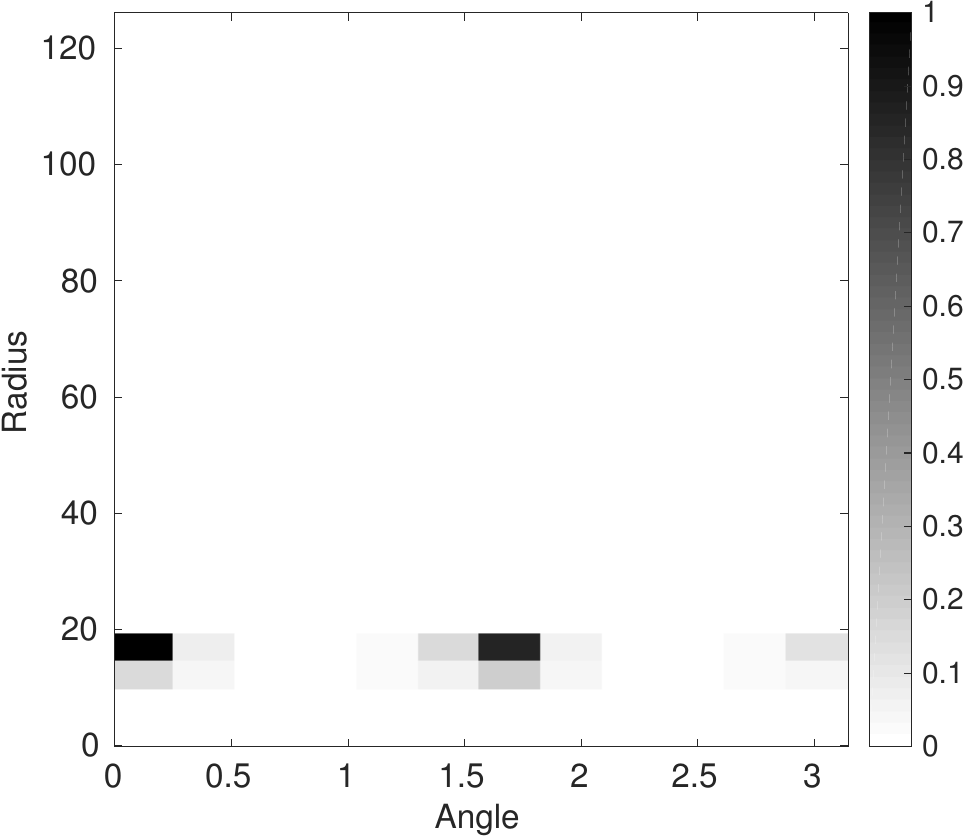}}& \raisebox{-0.05in} {\includegraphics[height=1in]{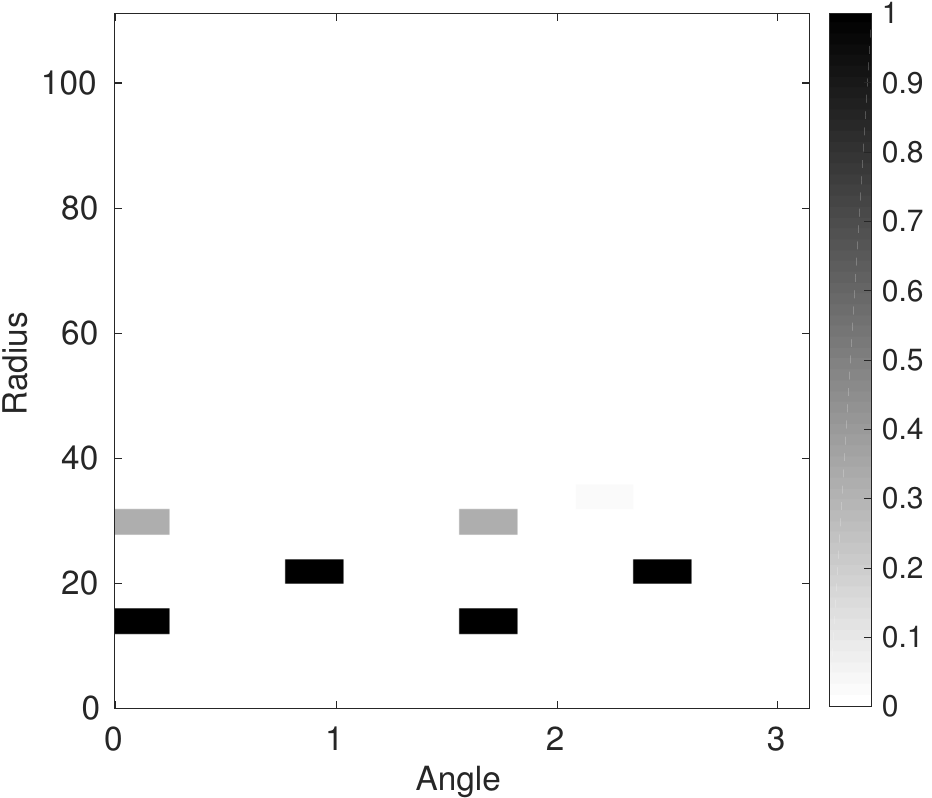} }
     \end{tabular}
  \end{center}
  \caption{Left: an example of atomic
  resolution crystal images $f(x) = \alpha(x)s(2\pi N\phi(x))$. Middle: the phase-space sketch $\tilde{S}(T^\cals_f,h,u)(r,\vartheta,x)$ by the polar windowed Fourier transform for a fixed $x$ (by a similar strategy of sketching applying on energy distribution of the polar windowed Fourier transform). Right: the phase-space sketch $\tilde{S}(T^\cals_f,h,u)(r,\vartheta,x)$ by the SST at the same point $x$. The sketch by the SST is cleaner and provides multi points in the reciprocal lattice $\mathcal{L}^\ast$ of the Bravais lattice $\mathcal{L}$ of the shape function $s(x)$.}
  \label{fig:skcomp}
\end{figure}

Let us recall the definition of the SS energy distribution in Equation \eqref{eq:SED}, written in polar coordinate as
\begin{multline}\label{eqn:psst}
T_{\!f}(r,\vartheta,x) \!=\!  \iint_{(a,\theta)\in D(x,\gamma)}  \!\!|W_{\!f}(a,\theta,x)|^2 
 \delta(v_{\!f}(a,\theta,x)\!-(r\cos \vartheta,r\sin \vartheta)) \, a\ud a\ud\theta\,,
\end{multline}
for $(r,\vartheta)\in[0,\infty)\times[0,\pi)$. Note that the SS energy distribution is translation invariant because the modulus of the wave packet coefficient $|W_f(a,\theta,x)|$ is invariant to the translation of the image $f(x)$. A simple idea to achieve the rotation invariance might be shifting $T_{\!f}(r,\vartheta,x)$ in the variable $\vartheta$ such that $T_{\!f}(r,\vartheta,x)$ always takes its maximum value at a specific location. However, this shifting procedure is usually sensitive to the crystal deformation. This motivates the following shifting and coarsening process to reduce the influence of rotation and deformation simultaneously. For a step size $h\in(0,\pi)$ in the angle coordinate $\vartheta$ and a step size $u>0$ in the radial coordinate $r$, we define the phase-space sketch via a coarsening procedure as follows.

Let \[(r_0(x),\theta_0(x))\in\arg\max_{r\in[0,\infty),\theta\in[0,\pi)} T_{\!f}(r,\vartheta,x),\]
we defined the shifted and rescaled SS energy distribution with a scaling parameter $\cals$ as 
\begin{equation}\label{eqn:pp}
T^\cals_{\!f}(r,\vartheta,x) = T_{\!f}\Bigl(\frac{r_0(x)r}{  \cals+u/2},\vartheta+\theta_0(x)-\frac{h}{2}\mod\ \pi,x\Bigr). 
\end{equation}
In the maximization above,  we approximate $[0,\infty)\times[0,\pi)$ with a finite set of points, numerically evaluate the objective at all these points, and identify the maximizer by simply traversing all the points. To reduce the computational cost for traversing, we will project the matrix in $\theta$ and $r$ into 1D vectors in $\theta$ and look for the peaks in $\theta$ first. Once the peaks in $\theta$ have been identified, it is only necessary to search for peaks in $r$ for $\theta$ around the peaks in $\theta$. This can reduce the computational complexity to a linear scaling one in terms of the number of grid points in $r$ and $\theta$. Typical choices are $30$ or $45$ grid points in both angle and radius domains for all numerical examples in this paper. Hence, the total computational cost for all $x$ is in the order of $10$ times the number of pixels. The step size $h$ is usually chosen to be $\frac{\pi}{9}$ or $\frac{\pi}{12}$, and the step size $u$ is usually chosen to be $\frac{1}{4}$ or $\frac{1}{2}$ of the lowest frequency of wave-like components in the crystal image, i.e., $N$ in our mathematical model in Equation \eqref{eqn:imagefunction}, which can be estimated by the SST by Theorem \ref{thm:main}. The same parameters will be used to discretize the angle and radius domain, and to do the sketching for other formulas in this paper since the proposed method is not sensitive to them. For a fixed $x$, there might be multiple maximizer but one can simply choose one of them. $T^\cals_{\!f}(r,\vartheta,x)$ is not sensitive to the choice of $(r_0(x),\theta_0(x))$ due to the concentration of the phase space representation and the largest spectrum energy of crystal images is assocated with the energy bump with the lowest wave number. 
Then the phase-space sketch of the new SS energy distribution
$T^\cals_{\!f}(r,\vartheta,x)$ is defined as a local average of the
shifted and rescaled SS energy distribution
\begin{equation}
\widetilde{S}(T^\cals_f,h,u)(r,\vartheta,x)\!=\! \int_{ \floor*{\frac{\vartheta}{h}}h}^{\floor*{\frac{\vartheta}{h}}h +h } \int_{ \floor*{\frac{r}{u}}u}^{\floor*{\frac{r}{u}}u+u  } T^\cals_{\!f}(\widetilde{r},\widetilde{\vartheta},x)  \,  \widetilde{r} \ud \widetilde{r}\ud\widetilde{\vartheta}\,
\end{equation}
for $(r,\vartheta)\in[0,\infty)\times[0,\pi)$, where $\floor*{\cdot}$ means the floor operator. 

This operation can be understood as a binning (or pooling) operator that collects
varying local wave vectors into fixed bins in the phase space, and
hence reduces the influence of the deformation.  Overall, the
phase-space sketching is a nonlinear transform (with respect to the
image $f$) that squeezes the phase-space energy distribution via
synchrosqueezing and coarsening into the sketch
$\widetilde{S}(T_f,h,u)$, resulting in a representation invariant to small deformation.

The extend of the deformation invariance of $\widetilde{S}(T^\cals_f,h,u)$ is
determined by the step sizes $h$ and $u$. Note that a fully
deformation invariant representation would not be able to distinguish
different Bravais lattices (see Figure \ref{fig:Bravais}), since these
lattices are the same in the quotient group of affine
transforms. Hence, we should choose appropriate step sizes $h$ and $u$
such that: 1) they are sufficiently large, making $\widetilde{S}(T^\cals_f,h,u)$
invariant to small elastic deformation due to external forces on the
material; 2) they are small enough such that $\widetilde{S}(T^\cals_f,h,u)$ is
capable of distinguishing different lattices.


The sketch $\widetilde{S}(T^\cals_f,h,u)$ is also invariant to
rotation and scaling. After shifting and rescaling, for a fixed $x$,
the new SS energy distribution $T^\cals_{\!f}(r,\vartheta,x)$ reaches
its maximum value at $(\cals+\frac{u}{2},\frac{h}{2})$. Hence, the
sketch $\widetilde{S}(T^\cals_f,h,u)(r,\vartheta,x)$ has the maximum
value in $[\cals,\cals+u)\times [0,h)$ for each fixed $x$.


It is worth pointing out that the sketch
$\widetilde{S}(T^\cals_f,h,u)(r,\vartheta,x)$ is separably invariant
with respect to rotation and translation. Therefore, the sketch cannot
discriminate a class of similar texture patterns, e.g. see Figure
\ref{fig:sim} for an example of two images sharing similar absolute
values of the wave packet coefficients. Fortunately, this is not a
problem in atomic resolution crystal image analysis, since all crystal
patterns behave similarly to the pattern in Figure \ref{fig:sim}
(left).

Finally, to remove the influence of the amplitude function $\alpha(x)$,
i.e., obtaining the
illumination invariance, we normalize the magnitude of the phase-space
sketch $\widetilde{S}(T^\cals_f,h,u)(r,\vartheta,x)$ and define
\begin{equation}\label{eqn:sk}
S(T^\cals_f,h,u)(r,\vartheta,x) =\dfrac{ \widetilde{S}(T^\cals_f,h,u)(r,\vartheta,x)}{\displaystyle \max_{r\in[0,\infty),\theta\in[0,\pi)}   \widetilde{S}(T^\cals_f,h,u)(r,\vartheta,x)}.
\end{equation}

 \begin{figure}
  \begin{center}
  \begin{tabular}{c}
        \includegraphics[height=1in]{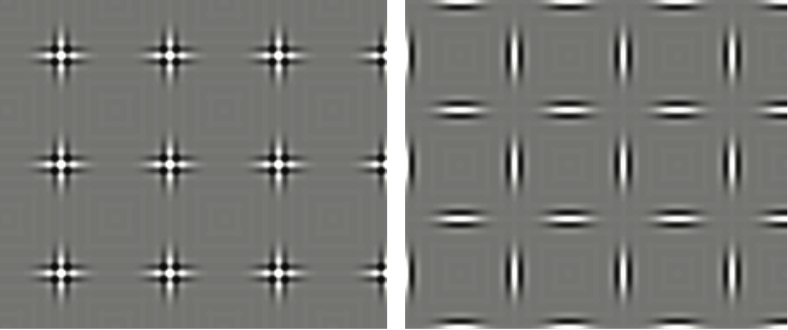}
     \end{tabular}
  \end{center}
  \caption{A separable invariant along rotations and translations cannot distinguish the left and right texture patterns, but a joint rotation-translation invariant can \cite{Sifre:2013}. Courtesy of Sifre et al..}
  \label{fig:sim}
\end{figure}

One may also consider sketching other more standard phase-space
representations, e.g., the polar windowed Fourier transform or wavelet
transform. However, it is more advantageous to combine sketching with the
synchrosqueezed transforms. As we have shown in Figure
\ref{fig:patch}, the SS energy distribution has concentrated support
around the local wave vectors of the crystal image, while the results
of windowed Fourier transform or wavelet transform are more
spread-out. Hence, the sketch of the windowed Fourier or wavelet
transform is not as clean as that of the SS energy distribution
(see the comparison in Figure \ref{fig:skcomp}).
Moreover, the SS
energy distribution provides useful information for crystal image
analysis and computing its sketching only adds minimal computational overhead. 
%
%

Finally, we close the introduction of phase space sketching with the following theorem.

\begin{theorem}
  \label{thm:main2}
  Suppose $f_1(x)=s(2\pi N_1x)$ and $f_2(x)=\alpha(x)s(2\pi N_2\phi(x))$ for $x\in[0,1]^2$ are two $2D$ GIMT's of type $(M,N_1)$ and $(M,N_2)$, respectively, such that $s(x)$ is a shape function describing one crystal configuration (square or hexegonal lattice). $\forall\epsilon>0$ and $n_0\in \mathbb{N}^+$, let $\gamma>\sqrt{\epsilon}$, compute the local wave vector estimations as in \eqref{eqn:lwv} and the SS energy distribution $T_{\!f_j}(r,\vartheta,x)$ as in \eqref{eqn:psst} in the domain $[0,n_0N_j]\times[0,\pi)\times [0,1]^2$ for $j=1$ and $2$. $\exists N_0(M,r,s,t,d,\eps)$ given by Theorem \ref{thm:main} such that $\forall N_1>N_0$, $\forall N_2>N_0$, and $\forall\gamma>\gamma_0$, if positive numbers $\rho$, $h$, and $u$ satisfy the following conditions:
  \[
  u>4R_{M,\epsilon} n_\gamma (\rho+\frac{u}{2}),\quad h>2\arcsin\left(2R_{M,\epsilon}\right),\quad\frac{\cals}{u}\in\mathbb{N}^+,\quad\frac{\pi}{6h}\in\mathbb{N}^+,
  \]
  where $R_{M,\epsilon}:=\max\{M-1,1-\frac{1}{M}\}+M\sqrt{\epsilon}$, $n_\gamma:=\min\{n\in\mathbb{N}^+:|W_{f_j}(a,\theta,b)|\leq \gamma,\forall a>nN_j \text{ and }\forall j\}$, then the phase space sketch of $f_1(x)$ and $f_2(x)$ satisfy
  \[
  S(T^\cals_{f_1},h,u)(r,\vartheta,x) =S(T^\cals_{f_2},h,u)(r,\vartheta,x) +O(\sqrt{\epsilon})
  \]
  for all $(r,\vartheta,x)\in [0,n_0\cals]\times[0,\pi)\times[0,1]^2$.

\end{theorem}
\begin{proof}
The proof of Theorem \ref{thm:main2} is straightforwad and only need elementary algebraic calculation. The intuition is
that when $\gamma$ is sufficiently large, we have thresholded less important features in the shape function
$s(x)$ and only a few important local wave vectors of crystal images are left for pooling, which allows
the choices of $\rho$, $u$, and $h$ such that the phase space sketch only contains a few nonzero energy
bins indicating the type of crystal configuration.  We will focus on the case of square lattice and use Figure  \ref{fig:p} to sketch out the proof. Let's assume $N_1 = N_2 = N$ without the loss of generality since we can rescale the SS energy distribution according to \eqref{eqn:pp}. Orange areas in Figure \ref{fig:p} (a) and (b) denote
the supports of the SS energy distribution $T_{f_1} (r, \vartheta, x)$ for some vector $n$ and for a fixed $x$; the green
areas denote the supports of $T_{f_2} (r, \vartheta, x)$ for the same $n$ and $x$; the blue areas cover the possible
supports of $T_{f_2} (r, \vartheta, x)$ for all $x$ with the fixed $n$. By Theorem \ref{thm:main}, we can estimate the size of these
colored areas up to a constant prefactor independent of $N$ (see quantitative estimations in Figure \ref{fig:p}
(b)). Hence, after shifting and rescaling, we obtain $T^\cals_{f_1} (r, \vartheta, x)$ in Figure \ref{fig:p} (c) and $T^\cals_{f_2} (r, \vartheta, x)$ in Figure \ref{fig:p} (d). Finally, as long as $u$ and $h$ are large enough such that the blue areas in Figure \ref{fig:p} (c) and (d) are
covered in a box given by the grids, then we see that the phase space sketches of $f_1$ and $f_2$ are the same up to $\sqrt{\epsilon}$ in the range $[0,n_0\rho]\times[0,\pi)\times[0,1]^2$. The conditions $\frac{\rho}{u}\in N^+$, and $\frac{\pi}{6h}\in N^+$ make sure that a single bin in the phase space sketches can cover one blue area in Figure \ref{fig:p} (c) and (d). 
\end{proof}

Theorem \ref{thm:main2} shows that for crystal images with the same type of configuration (either square or hexagonal lattice), their phase space sketches are invariant to image translation, rotation, illumina- tion, and scale difference. This new representation is also stable to deformation and invariant to a class of elastic deformation characterized by the phase function $\phi(x)$ in 2D GIMTs. When the parameter $M$ in a GIMT is $O(1)$, it is easy to find parameters $\cals>0$, $h>0$, and $u>0$ to construct the phase space sketch. However, if $M$ is large, there is no good parameter to obtain invariant representations via phase space sketching. 

 \begin{figure}
  \begin{center}
  \begin{tabular}{c}
        \includegraphics[height=4.2in]{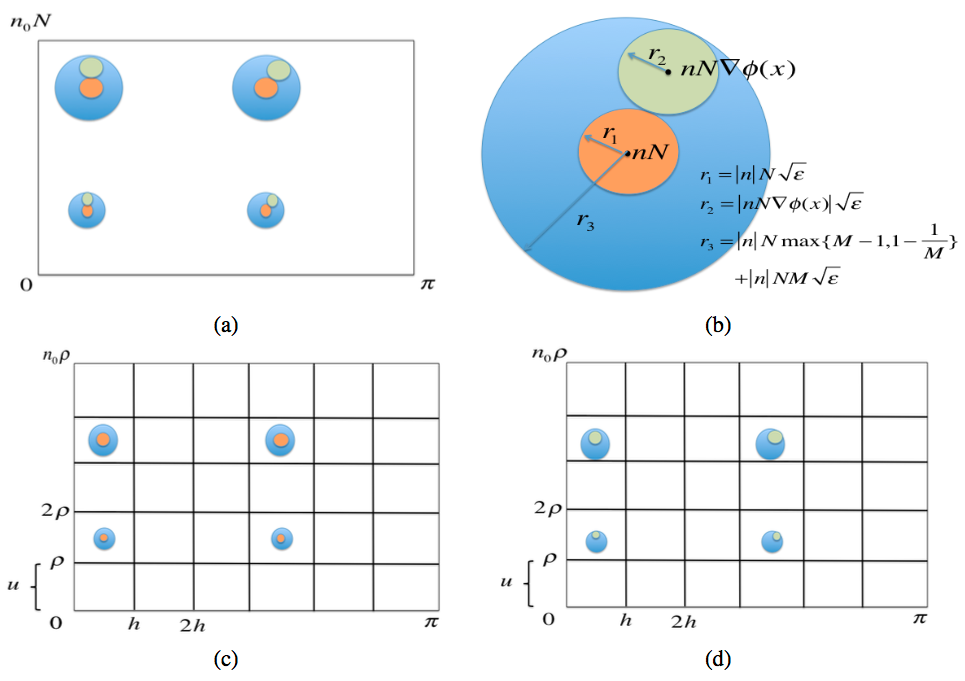}
     \end{tabular}
  \end{center}
  \caption{Illustration of the proof of Theorem \ref{thm:main2} in the case of square lattice in the asymptotic sense, i.e., the equalities and inequalitys hold up to a constant independent of $N$. (a) Supports of the SS energy of $f_1$ (in orange) and $f_2$ (in green) in
polar coordinate. Blue areas indicate the possible locations of the supports of the SS energy of $f_2$ when $\phi$ varies. (b) A zoomed-in example of the supports in (a) with
detailed locations in Cartesian coordinate for an $n$ such that $|n| \leq n_0$. (c) $T^\cals_{f_1} (r, \vartheta, x)$.
(d) $T^\cals_{f_2} (r, \vartheta, x)$.}
  \label{fig:p}
\end{figure}

\subsection{Classification}
\label{sub:cls}

We may apply the phase-space sketch to classify crystal textures in a
complicated crystal image. Although advanced classifiers such as SVM are useful tools
for classification, here we present some simple and efficient
classification methods.

As we have seen in Figure \ref{fig:skcomp}, the phase-space sketch
$S(T^\cals_f,h,u)$ at a point $x$ in the interior of a grain has
well-separated supports indicating the reciprocal lattice
$\mathcal{L}^\ast$ of the Bravais lattice $\mathcal{L}$ of the crystal
pattern. Hence, the locations and the magnitudes of the supports are
important features for crystal pattern classification. In the case of
a complicated crystal image, if two grains share the same crystal
pattern, they should have the same phase-space sketch
$S(T^\cals_f,h,u)$ at the locations $x$ sufficiently far away from
grain boundaries. Therefore, we only need to identify major groups of
the sketch $S(T^\cals_f,h,u)$ at different $x$'s and choose a
representative sketch from each major group. Other sketch outliers are
due to the influence of grain boundaries on the phase-space
representation; each sketch outlier contains the information of at
least two crystal patterns and looks like a superposition of more than
one sketches.


 \begin{figure}
  \begin{center}
  \begin{tabular}{cc}
        \includegraphics[height=1in]{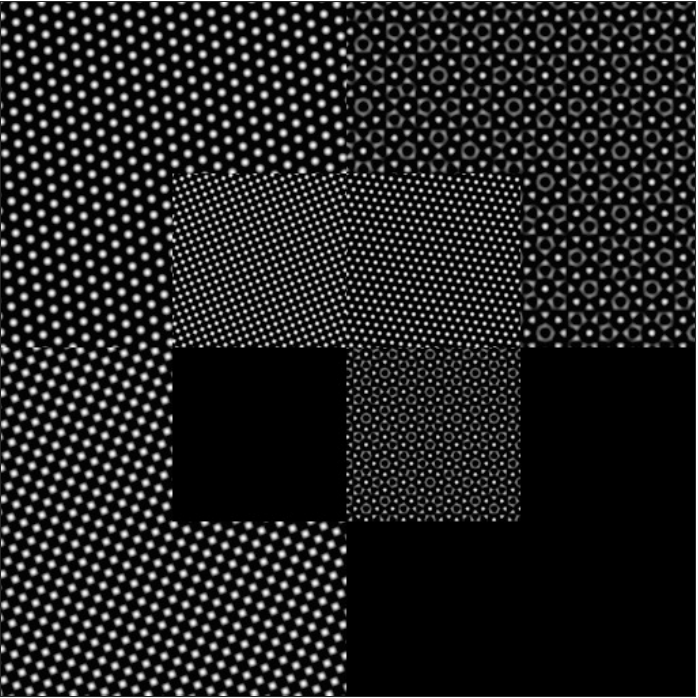}&  \raisebox{-0.02in} {\includegraphics[height=1.025in]{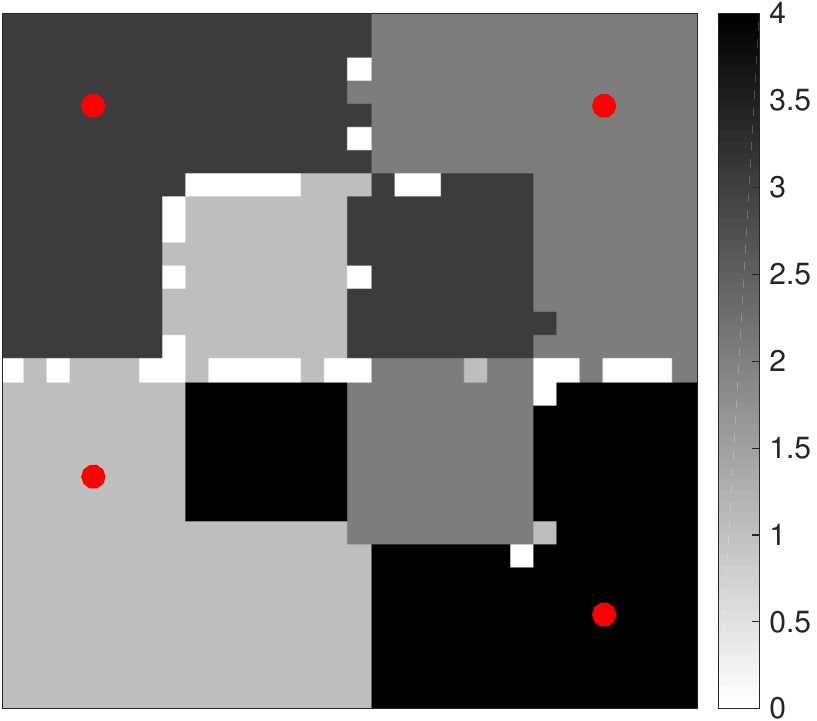}}
     \end{tabular}
  \end{center}
  \caption{Left: a toy example of atomic resolution crystal images with different crystal patterns. Right: classification results of the sketch by the SST. The crystal image is of size $512\times 512$. We sub-sample the crystal image to save computation (one sample every $4$ pixels) and generate local patches of size $65\times 65$ pixels centered at the sub-sampled points. Each patch is associated with a phase-space sketch. These patches are classified based on the compressed feature vector of the sketch using spectral clustering. The centers of outlier patches are indexed with zeros.}
  \label{fig:cls2}
\end{figure}

A simple idea for classification is to define an appropriate distance
to measure the similarity of different phase-space sketches and apply
the spectral clustering \cite{Ng:2001} to identify the major groups of the sketches. 
%
%
Since the types of crystal textures are limited, it is not necessary to use the whole sketch for discrimination. To ensure that the method is as efficient as possible, features contained in the sketch  invariant to uninformative variability in crystal images are more important. Observe that the phase-space sketch is able to sketch out the multi-scale reciprocal lattice $\mathcal{L}^\ast$ (see Figure \ref{fig:skcomp} (right)), and the numbers of supports at different scales largely determine the crystal pattern. It is sufficient to use these numbers as a compressed feature vector to represent the sketch. For example, the feature vector of the sketch in Figure \ref{fig:skcomp} (right) is $(2,2,2)$, where each $2$ means that there are $2$ leading supports above a certain threshold at each scale (each row in the phase space sketch matrix). For some other examples for the feature vectors, please see Figure \ref{fig:cls2rs}. Note that the number of supports are invariant to translation, rotation, illumination, and small deformation. Hence, the feature vector by the number of supports in sketches is a compressed invariant representation of crystal patterns. The standard Euclidean distance of vectors is a natural choice to measure the similarity of these compressed invariant representations. Using the number of supports only might be too ambitious in some cases. A better way is to take advantage of the magnitudes of the peaks in these supports, once two crystal patterns have sketches sharing the same numbers of supports. For example, a simple idea is to set up a threshold parameter and only count the number of supports above this threshold. In practice, this simple idea is sufficient to discriminate most crystal patterns in real applications.

 \begin{figure}
  \begin{center}
  \begin{tabular}{cccc}
        \includegraphics[height=1in]{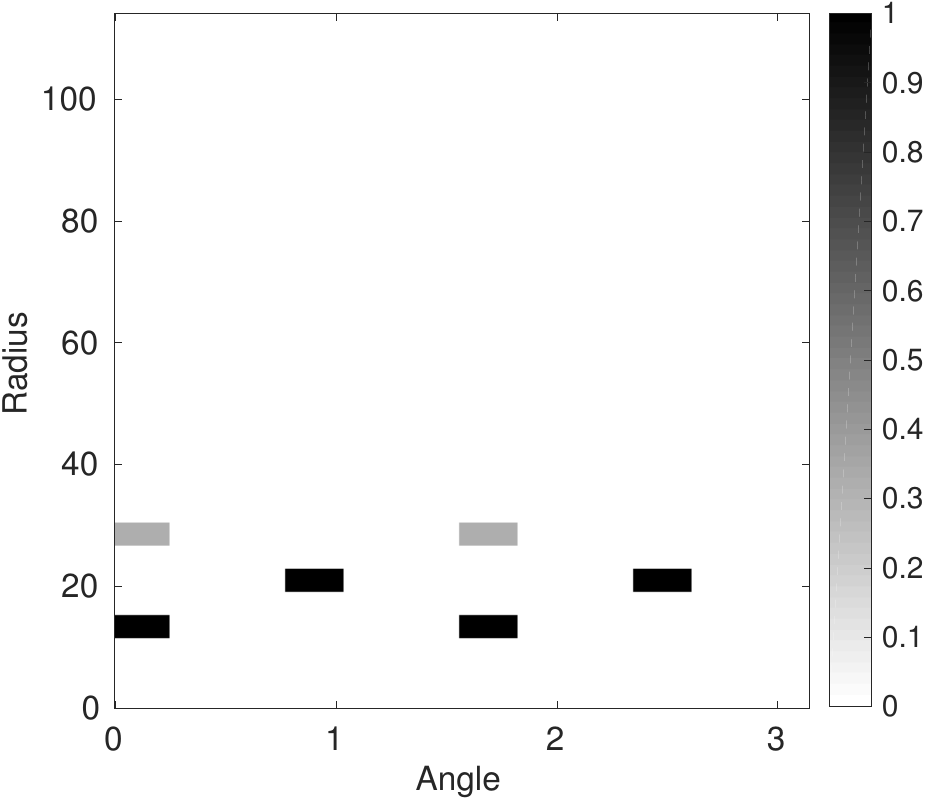} &\includegraphics[height=1in]{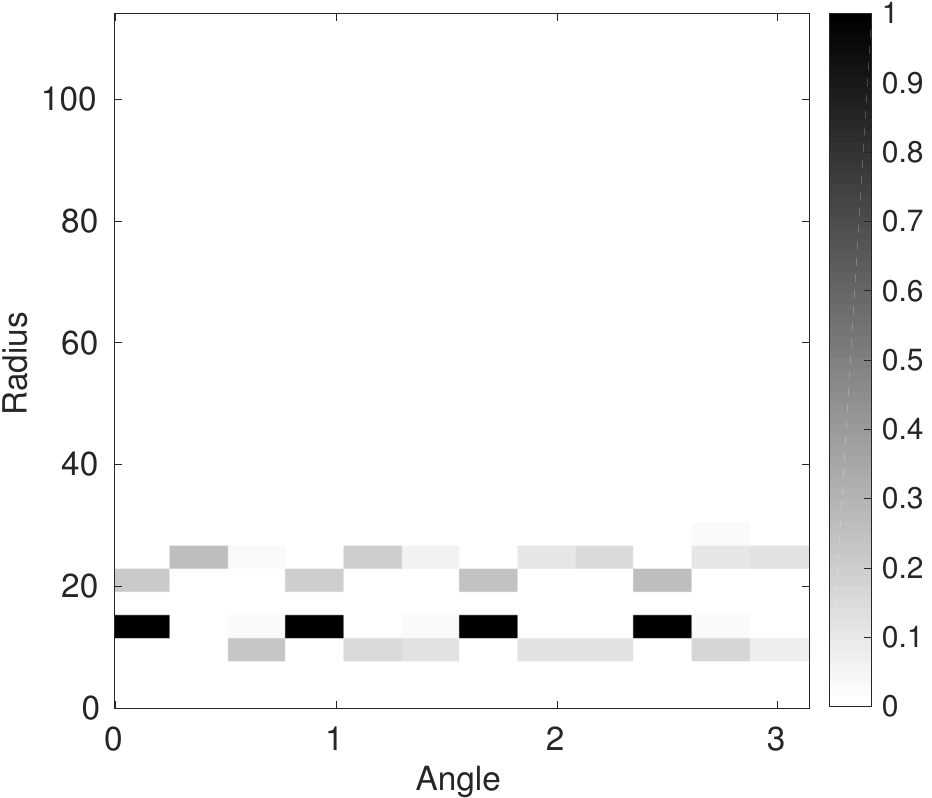} & \includegraphics[height=1in]{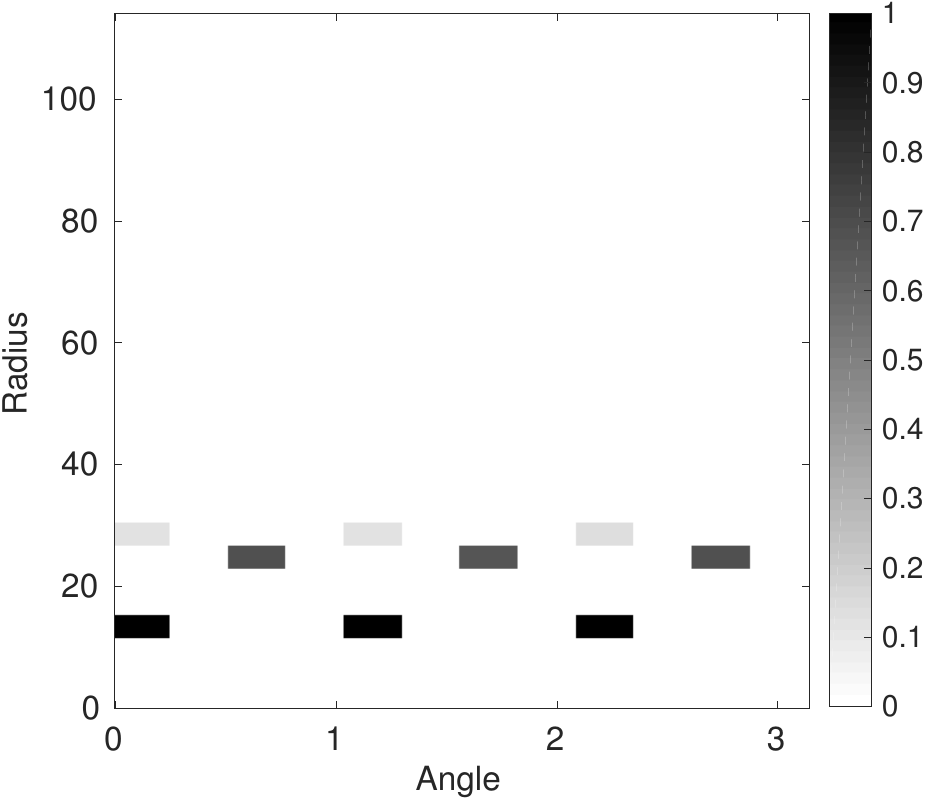}& \includegraphics[height=1in]{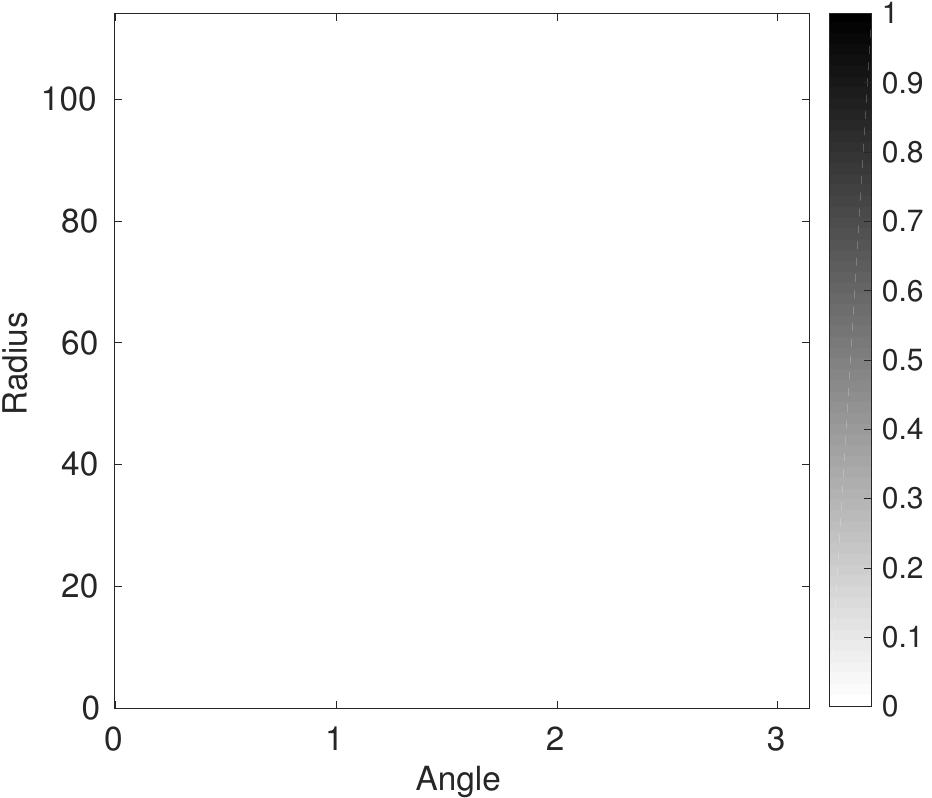} \\
        Type $1$ & Type $2$ & Type $3$ & Type $4$
     \end{tabular}
  \end{center}
  \caption{From left to right, the phase-space sketch $S(T^\cals_f,h,u)(r,\vartheta,x)$ of the SST corresponding to the type of reference configuration identified in Figure \ref{fig:cls2}. Red dots in Figure \ref{fig:cls2} (right) indicate the position $x$'s of the sketches in this figure. After thresholding these phase space sketches with a threshold $0.05$, the feature vectors, i.e., vectors describing the number of leading isolated supports, of these sketches are $(2,2,2)$, $(4,4,4,4)$, $(3,3,3)$, and an empty vector (from left to right).}
  \label{fig:cls2rs}
\end{figure}

Finally, we provide several synthetic examples to demonstrate the
efficiency of the proposed compressed feature vector for
classification. The first example in Figure \ref{fig:cls2} (left)
contains three different crystal patterns and one vacancy area, and 
each pattern has two grains with different orientations and scales. The
SST is applied to generate the
phase-space sketch and the corresponding compressed feature vectors at each
pixel of the crystal image. The spectral clustering algorithm is able
to identify four major groups of compressed feature vectors. According
to the clustering results of the feature vectors, we group the
corresponding pixels together and visualize the results in Figure
\ref{fig:cls2rs} (right). 

%

 \begin{figure}
  \begin{center}
  \begin{tabular}{cc}
        \includegraphics[height=1in]{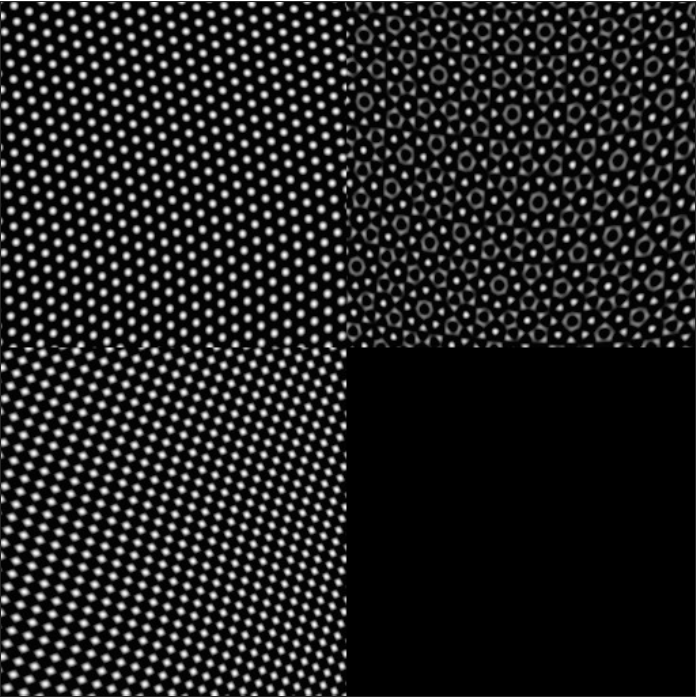}& \raisebox{-0.015in} {\includegraphics[height=1.025in]{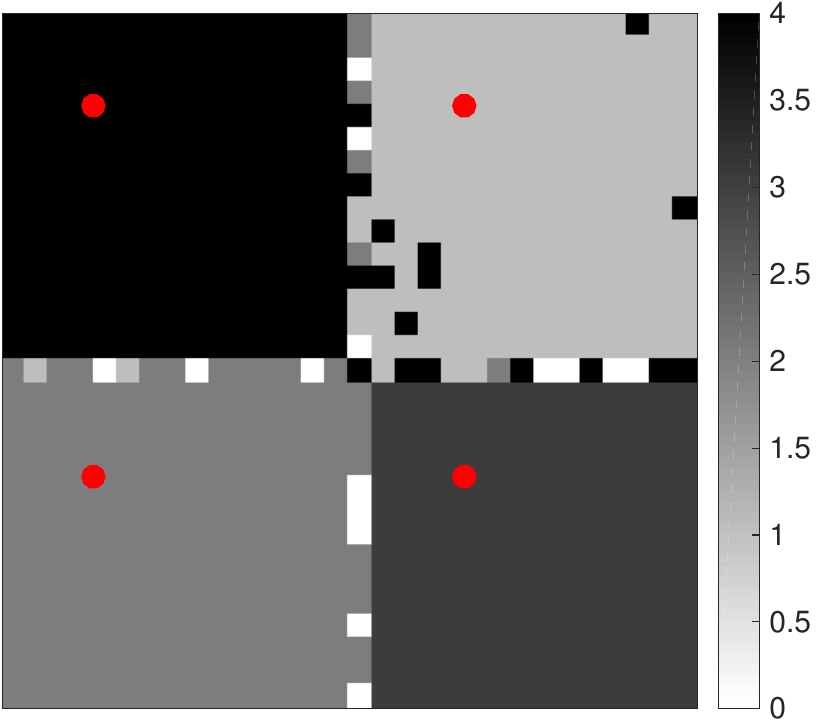}}
     \end{tabular}
  \end{center}
  \caption{Left: a toy example of atomic resolution crystal images
    with different crystal patterns. Right: classification results by
    the sketch of the SST. Numerical results were obtain using the
    same setting as in Figure~\ref{fig:cls2}.}
  \label{fig:cls3}
\end{figure}

 \begin{figure}
  \begin{center}
  \begin{tabular}{cccc}
        \includegraphics[height=1in]{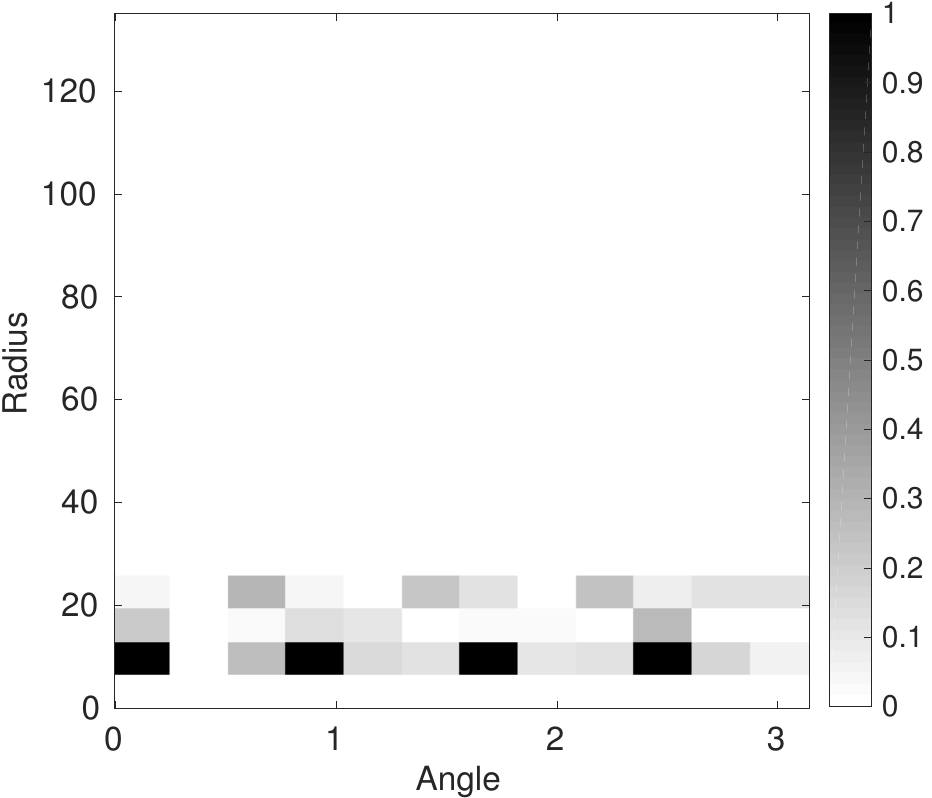}& \includegraphics[height=1in]{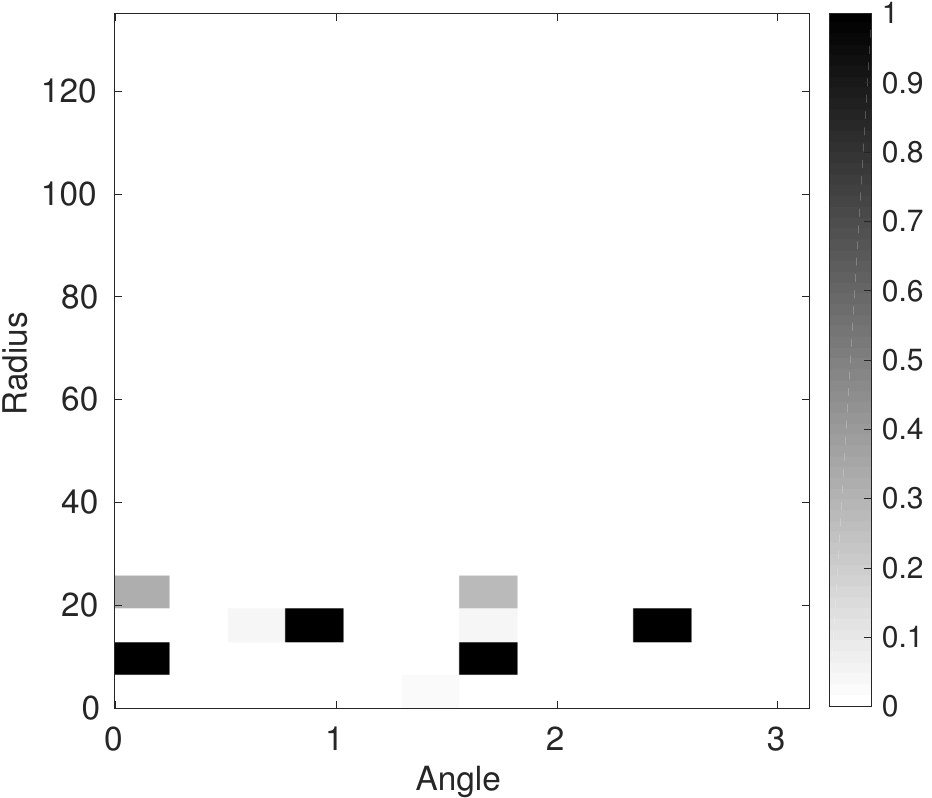} & \includegraphics[height=1in]{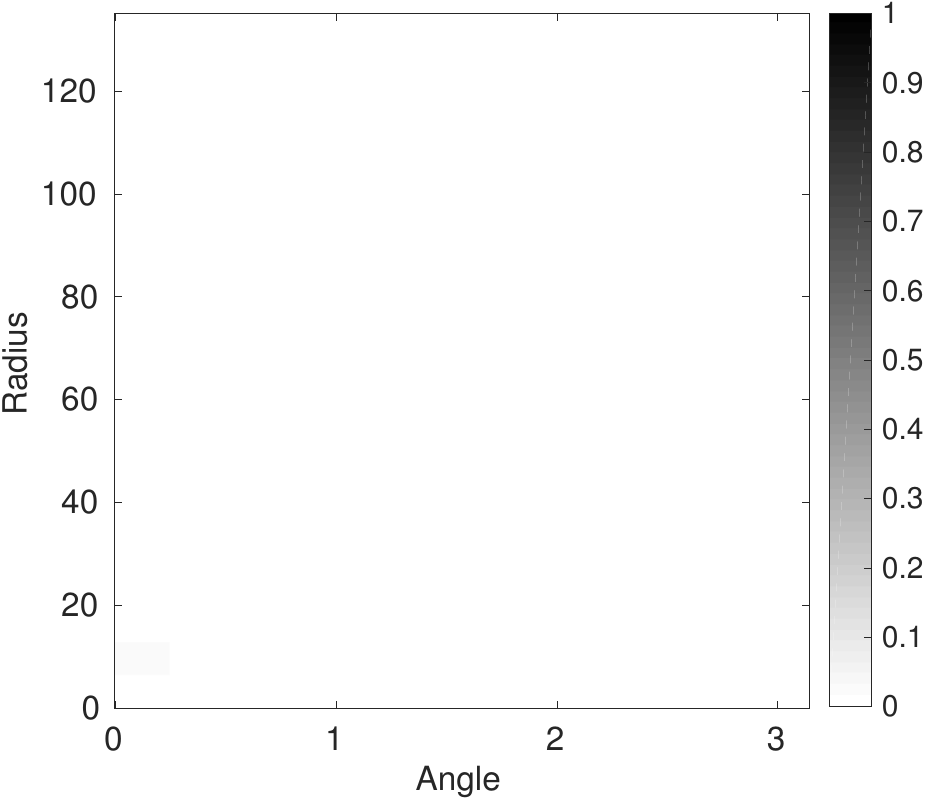} &\includegraphics[height=1in]{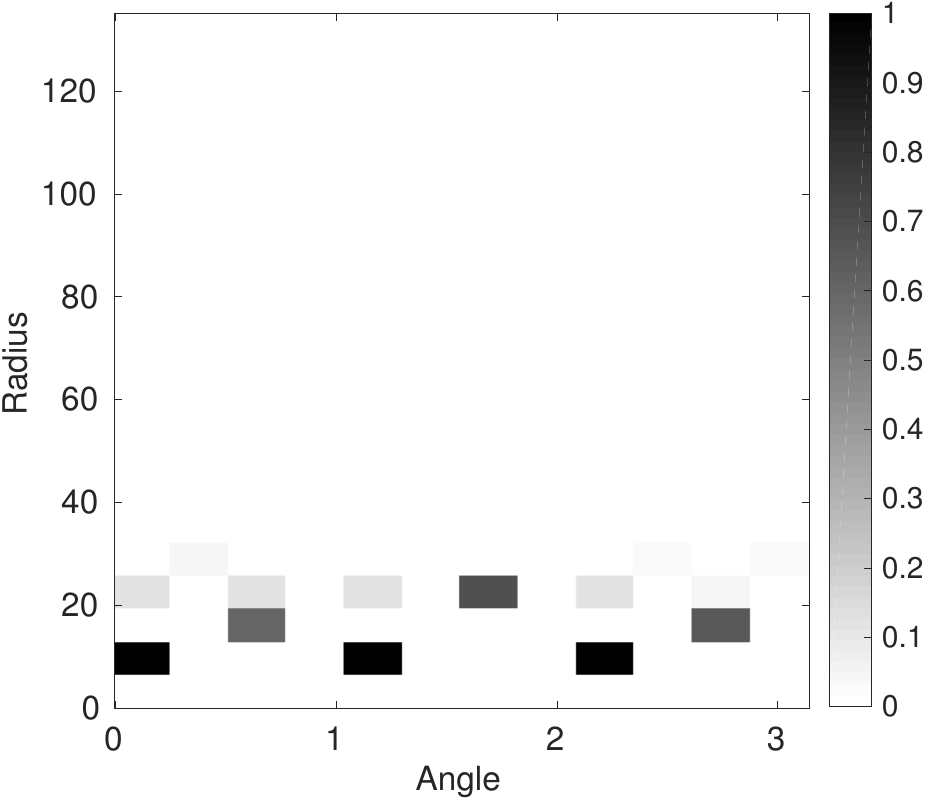} \\
        Type $1$ & Type $2$ & Type $3$ & Type $4$
     \end{tabular}
  \end{center}
  \caption{From left to right, the phase-space sketch $S(T^\cals_f,h,u)(r,\vartheta,x)$ of the SST corresponding to the type of reference configuration identified in Figure \ref{fig:cls3}. Red dots in Figure \ref{fig:cls3} (right) indicate the position $x$'s of the sketches in this figure.}
  \label{fig:cls3rs}
\end{figure}

The third example in Figure \ref{fig:cls3} (left) is similar to the first example but the crystal image has been smoothly deformed. As shown in Figure \ref{fig:cls3} (right) and Figure \ref{fig:cls3rs}, the proposed method succeeds to detect all crystal patterns and their sketches.

In the last example in Figure \ref{fig:cls4} (left), there are four different crystal patterns (two kinds of square lattices and two kinds of hexagonal lattices) with different levels of illumination. As shown in Figure \ref{fig:cls4} (right) and Figure \ref{fig:cls4rs}, the proposed method is able to detect all crystal patterns and their sketches. This is a very challenging example. As we can see in the sketches in Figure \ref{fig:cls4rs}, Type $1$ and $2$ sketches almost share the same support (and so do Type $3$ and $4$). Hence, as discussed previously, only the number of supports with peaks over an appropriate threshold is used in constructing the compressed feature vector.

 \begin{figure}
  \begin{center}
  \begin{tabular}{cc}
        \includegraphics[height=1in]{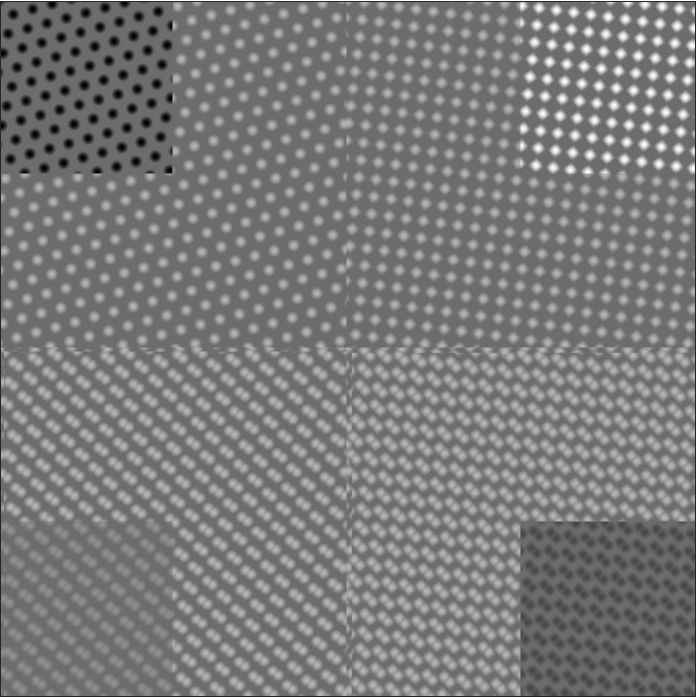}& \raisebox{-0.015in} {\includegraphics[height=1.025in]{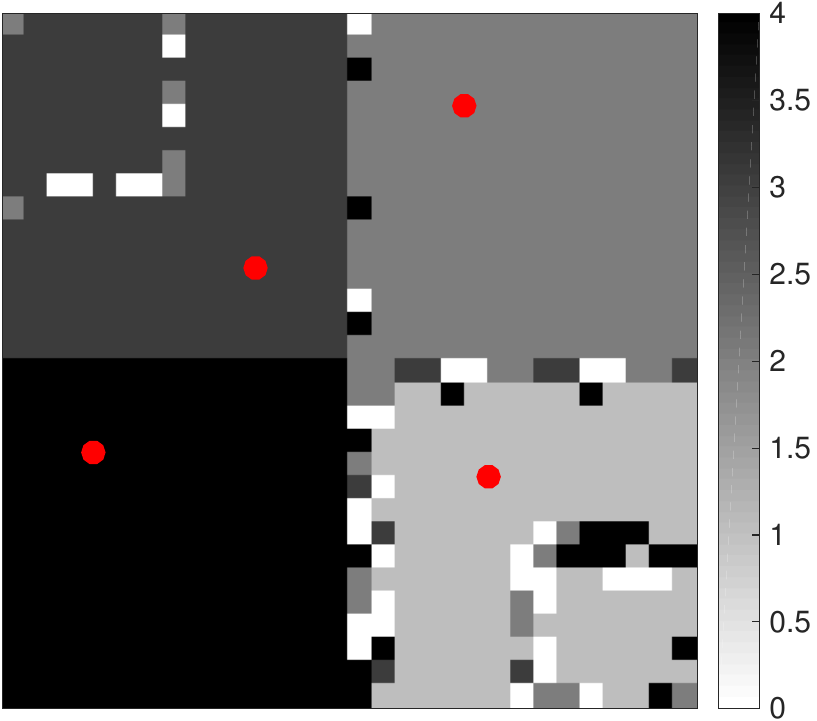}}
     \end{tabular}
  \end{center}
  \caption{Left: a toy example of atomic
  resolution crystal images with different crystal patterns. Right: classification results by the sketch of the SST. Numerical results were obtain using the same setting as in Figure~\ref{fig:cls2}.}
  \label{fig:cls4}
\end{figure}

 \begin{figure}
  \begin{center}
  \begin{tabular}{cccc}
        \includegraphics[height=1in]{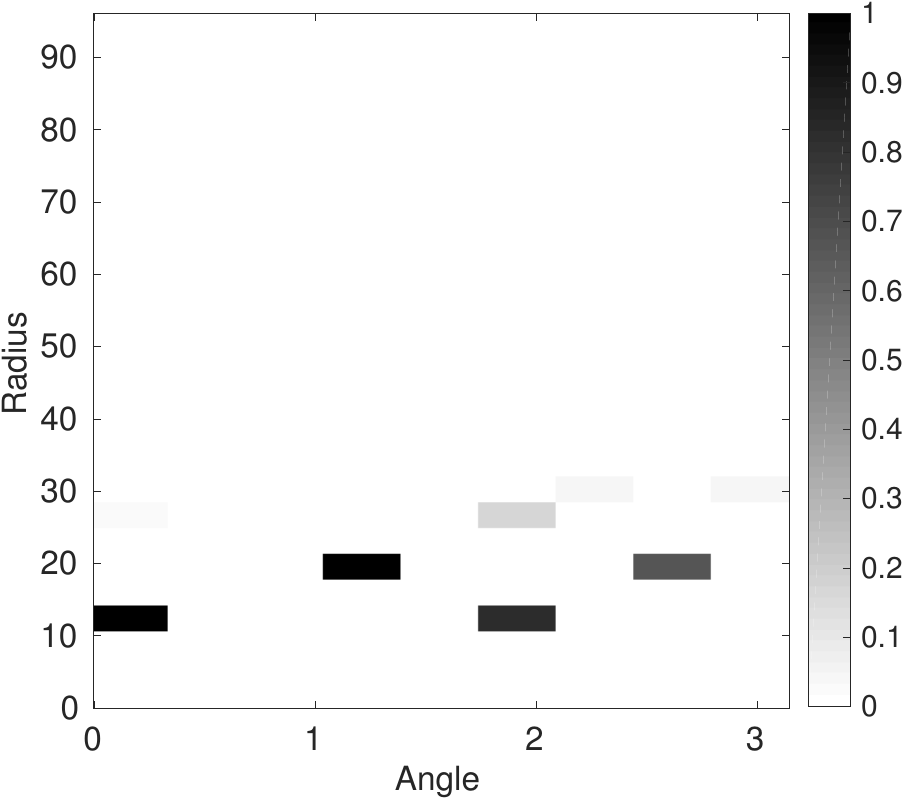}& \includegraphics[height=1in]{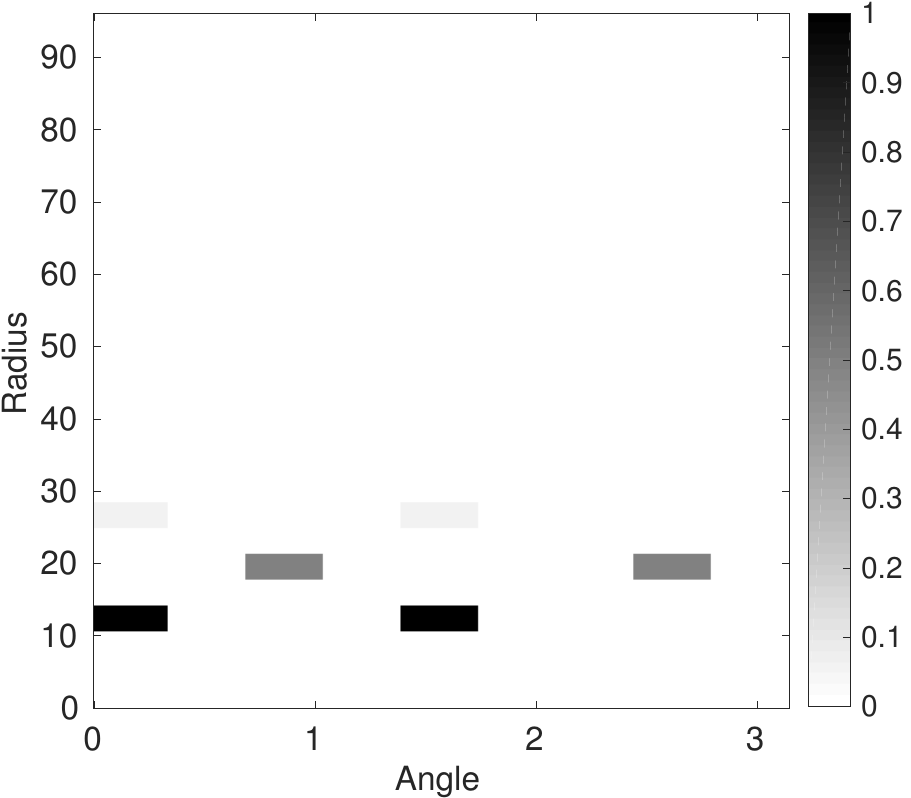}  &\includegraphics[height=1in]{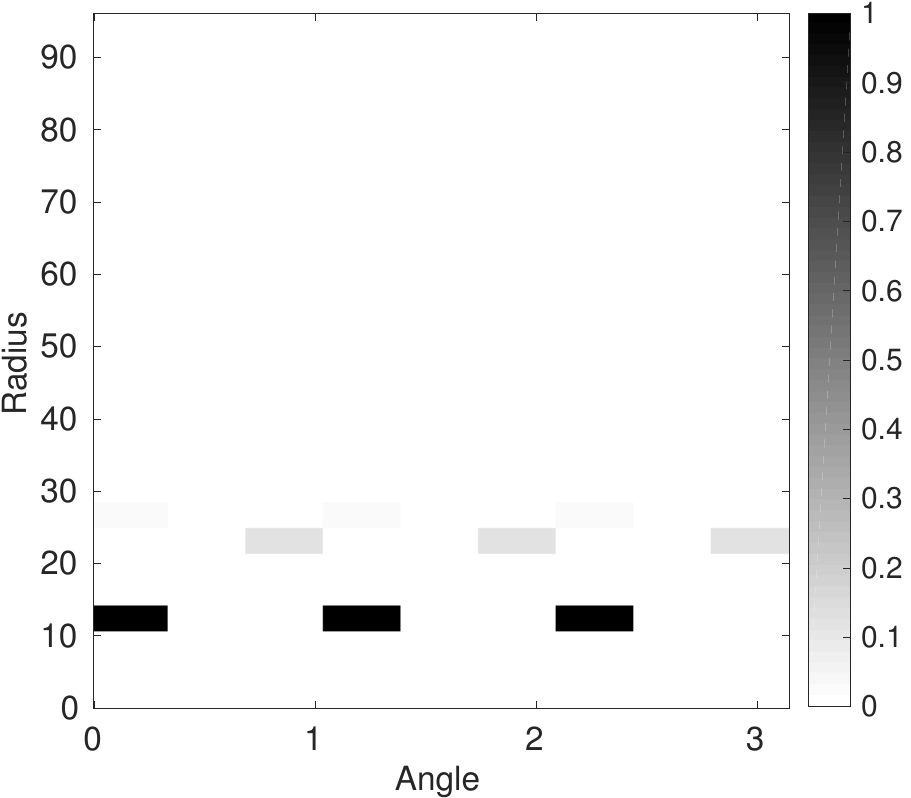} &\includegraphics[height=1in]{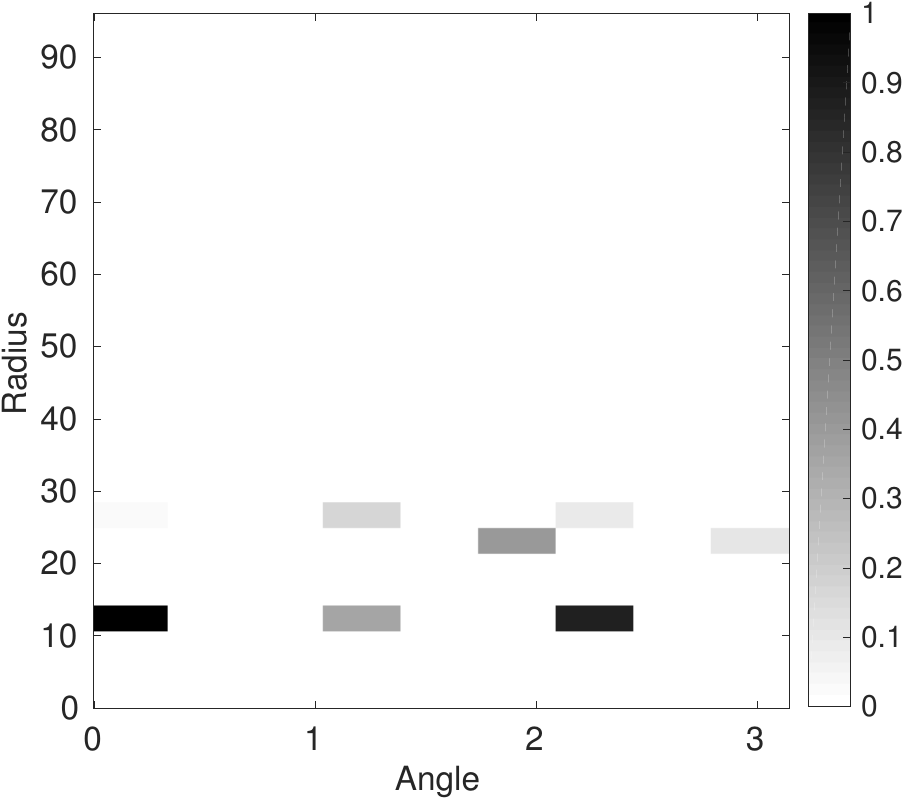} \\
        Type $1$ & Type $2$ & Type $3$ & Type $4$
     \end{tabular}
  \end{center}
  \caption{From left to right, the phase-space sketch $S(T^\cals_f,h,u)(r,\vartheta,x)$ of the SST corresponding to the type of reference configuration identified in Figure \ref{fig:cls4}. Red dots in Figure \ref{fig:cls4} (right) indicate the position $x$'s of the sketches in this figure.}
  \label{fig:cls4rs}
\end{figure}

To demonstrate the efficiency of the phace-space sketching method, we quantitatively compare it with the state-of-the-art algorithms\footnote{In our numerical tests, CNN \cite{CPA:CPA21413} doesn't perform well because crystal image data are limited and it is difficult to regularize the CNN with sufficiently good data augmentation. Hence, we only compare our algorithm with the scattering transform. Default parameters were used in the scatterign transform.}, e.g. the scattering transform in \cite{Sifre:2013}, using the examples in Figure \ref{fig:cls3} and \ref{fig:cls4}. Each example contains four different kinds of crystal texture; correspondingly, $36$ training images of size $64$ by $64$ with different orientation, translation, image intensity, deformation were generated for each texture class. Crystal images on the left of Figure \ref{fig:cls3} and \ref{fig:cls4} were partitioned into $32^2$ overlapping patches of size $64\times 64$ uniformly distributed in the image domain and these patches are used as validation data. The scattering transform was used to train and classify these data.  Table \ref{tab:comp1} and \ref{tab:comp2} summarize the performance of the propose method and the scattering transform. Numerical results show that the phase space method is highly efficient in running time and the success rate is higher than the scattering transform. The code and the data set for this comparison are available in the SynCrystal package. The numerical tests were performed on a MacBook Pro (processor 3.1 GHz Intel Core i5, memory 8 GB 2133 MHz LPDDR3).

\begin{table}[h]
\centering
\begin{tabular}{rcccc}
\toprule
   Example & Algorithm & $T_{sct}(sec)$ & $T_{cls}(sec)$ & Success Rate \\
\toprule 
   Figure \ref{fig:cls3} & non-log & 8.65e+02 & 2.60e-02 & 84.6  \\
   Figure \ref{fig:cls3} & log & 8.80e+02 & 3.02e-02 & 88.3  \\
\toprule 
   Figure \ref{fig:cls4} & non-log & 8.78e+02 & 2.60e-02 & 82.0  \\
   Figure \ref{fig:cls4} & log & 8.95e+02 & 2.60e-02 & 82.8  \\
\bottomrule
\end{tabular}
\caption{Performance of the scattering transform for crystal image classification. ``non-log" means the scattering transform without log scattering, while ``log" means with log scattering. $T_{sct}$ is the time for computing the scattering representation of all the training and validation data. $T_{cls}$ is the time for classifying the validation data. }
\label{tab:comp1}
\end{table}

\begin{table}[h]
\centering
\begin{tabular}{rcccc}
\toprule
   Example & $T_{sst}(sec)$ & $T_{skt}(sec)$ & $T_{cls}(sec)$ & Success Rate \\
\toprule 
   Figure \ref{fig:cls3} & 2.46e+01& 2.54e+01 & 1.80e+00 & 94.2  \\
\toprule 
   Figure \ref{fig:cls4} & 1.78e+01 & 2.16e+01& 1.91e+00 & 90.7  \\
\bottomrule
\end{tabular}
\caption{Performance  of the phace-space sketching algorithm for crystal image classification. $T_{sst}$ is the time for the synchrosqueezed transform; $T_{skt}$ is the time for phase-space sketching; $T_{cls}$ is the time for classification. The success rate has been visualized on the right of Figure \ref{fig:cls3} and Figure \ref{fig:cls4}.}
\label{tab:comp2}
\end{table}

\subsection{Segmentation}
\label{sub:seg}

In the previous section, we have applied the phase-space sketching to identify reference sketches for different crystal patterns in a complicated crystal image. As shown in Figure \ref{fig:cls2} to \ref{fig:cls4}, a crystal image can be roughly partitioned into several pieces and each piece contains grains sharing the same crystal pattern. However, we are not able to determine an accurate partition due to sketch outliers at grain boundaries. Since we already have the reference sketches by the algorithm in the previous section, we can simply match the sketch outliers with the reference sketches and identify the most possible type of crystal pattern. However, we do not expect the outlier treatment to give meaningful results at grain boundaries, because the results will hesitate between two types and in real data it is even difficult to distinguish types by manual inspection. Once we have completed the image segmentation for different crystal patterns, the atomic
  resolution crystal image analysis in Section \ref{sec:model} is applied within each type of crystal patterns to estimate defects and crystal deformations. This completes the complicated crystal image analysis. 

%

\section{Quantitative analysis in real applications}
\label{sec:analysis}

This section presents several real examples in materials science to demonstrate the efficiency of the proposed method in this paper. These atomic crystal examples include crack propagation, phase transition, and self-assembly.

 \begin{figure}
  \begin{center}
  \begin{tabular}{cc}
        \includegraphics[height=1.2in]{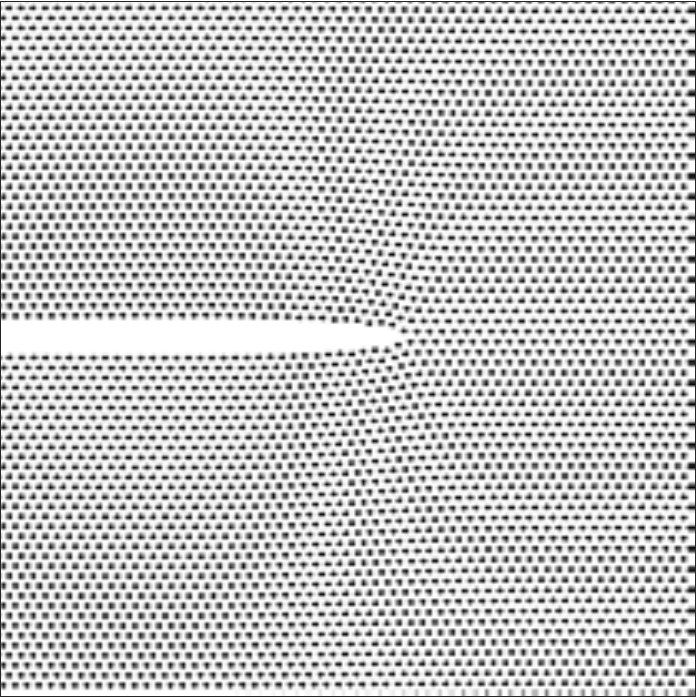} &     \includegraphics[height=1.2in]{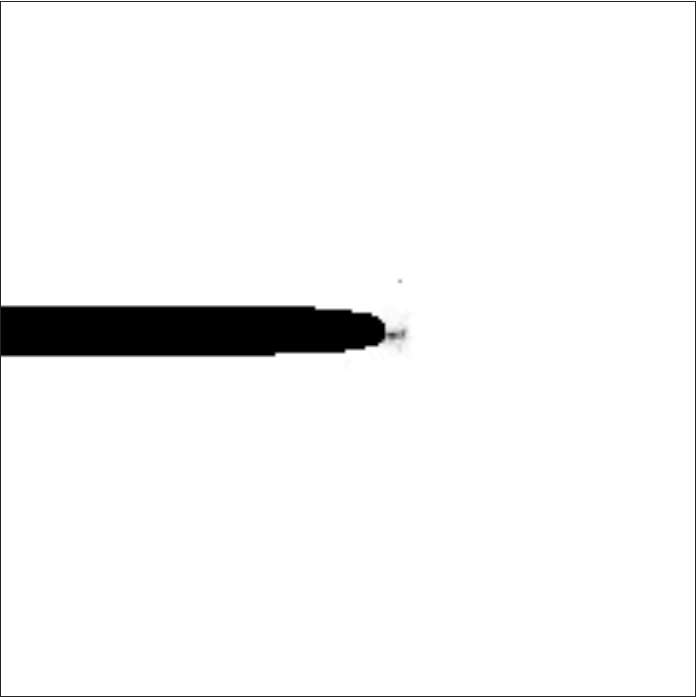} 
     \end{tabular}
       \end{center}
  \caption{Left: The atomic simulation at zero temperature to examine dislocation emission from a crack tip. Right: grain boundary by thresholding the defect indicator function of the left figure with a black mask for the vacancy area given by the classification using phase space sketching.}
  \label{fig:ex1}
\end{figure}

 \begin{figure}
  \begin{center}
       \begin{tabular}{ccc}
     &\includegraphics[height=1.2in]{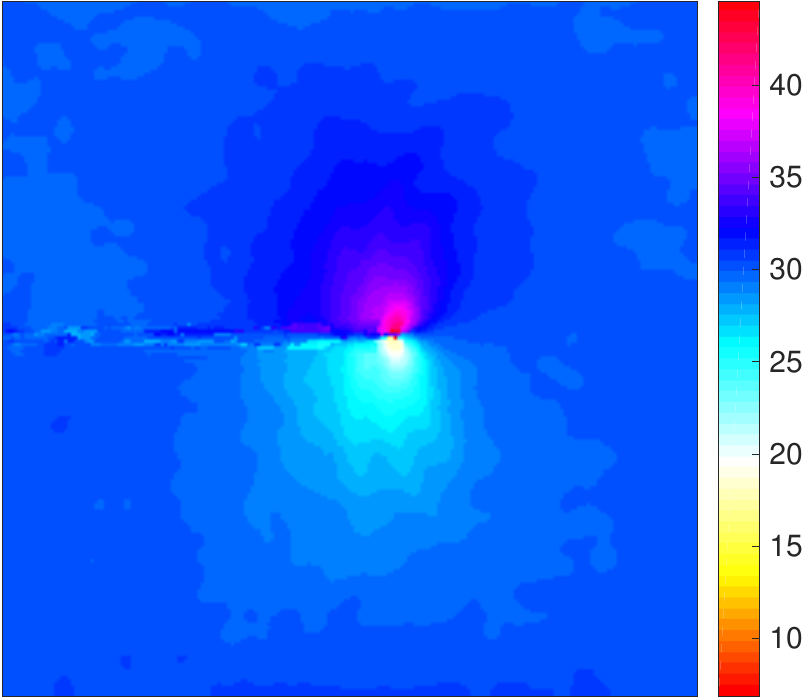}    \includegraphics[height=1.2in]{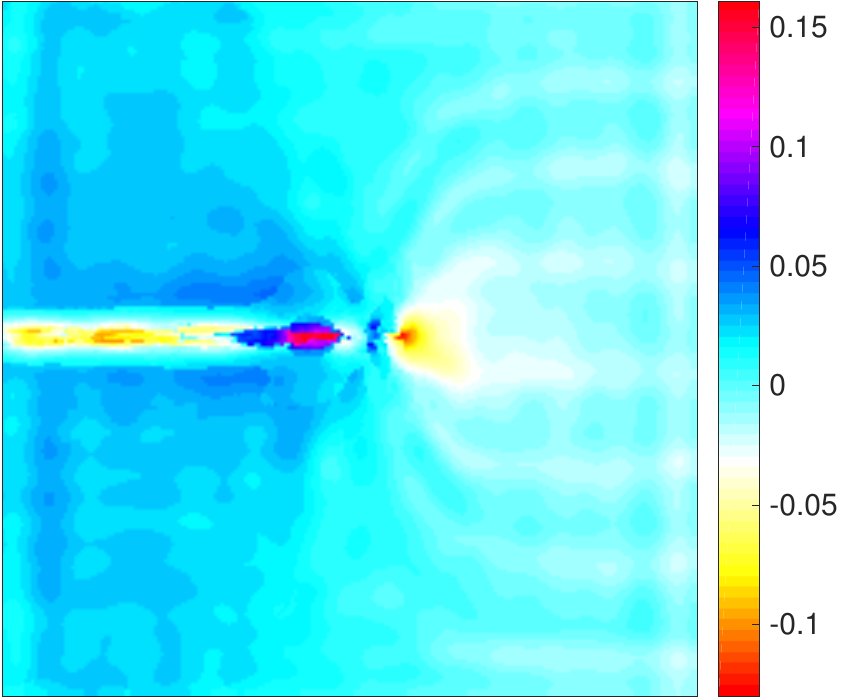} &\includegraphics[height=1.2in]{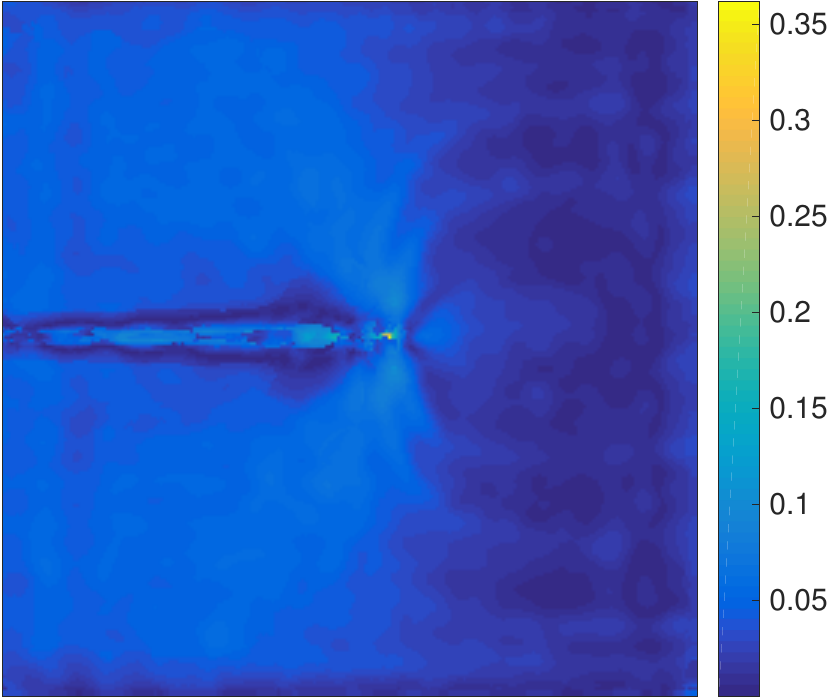}
     \end{tabular}
  \end{center}
  \caption{Numerical results of the example in Figure \ref{fig:ex1}. Left: crystal orientation. Middle: the volume distortion of the elastic deformation. Right: the difference in principle stretches of the elastic deformation.}
  \label{fig:rs1}
\end{figure}

\subsection{Crack  propagation}


Accurate and quantitative estimation of the elastic deformation in the
presence of plastic deformation and defects (cracks and dislocations)
is crucial for predicting a critical loading state that results in
crack growth and dislocation motion, understanding the dislocation
nucleation/emission from an originally pristine crack tip, quantifying
the effect that pre-existing dislocations surrounding a crack tip can
have on its driving force, and exploring how the crack-tip shape
impacts the dislocation emission process \cite{Zimmerman2}, just to
name a few.

Figure \ref{fig:ex1} (left) is an example of the atomic simulation at zero temperature to examine dislocation emission from a crack tip. We apply the proposed algorithms in this paper to identify the crystal pattern, the crack region, and estimate the elastic deformation of this example. As shown in Figure \ref{fig:ex1} (right), the crack estimation is accurate up to an error less than one atom. Hence, the proposed method is able to distinguish crystal configuration and vacancy area as in the toy examples in Section \ref{sec:sk}.  Figure \ref{fig:rs1} shows the estimation of the crystal orientation, the volume distortion and the difference in principle stretches of the elastic deformation. These results quantitatively show the interaction between the crack and the perfect crystal region through elastic deformation, especially around the crack-tip and in the direction of the crack propagation.

\subsection{Phase transition}

With the development of digital video microscopy \cite{melting2,melting5,Peng:2014}, researchers are able to observe the phase transition (between solid, liquid and gaseous states of matter) with single particle resolution. The behavior of important quantities at the atomic scale (e.g., defects, deformation, phase interfaces) presents interesting new questions and challenges for both theory and experiment in materials science. Research in this direction is limited by the difficulty of imaging and analyzing atomic crystals. In the aspect of data analysis, the challenge comes from the fact that it is difficult to track the atoms in the evolution process of phase transition, especially in the case of irregular patterns (e.g. liquid or gas states). The proposed algorithm in this paper is free of tracking atoms, offering a new means to examine fundamental questions in phase transition.

\subsubsection{Solid-liquid phase transition}

 \begin{figure}
  \begin{center}
  \begin{tabular}{cc}
        \includegraphics[height=1.2in]{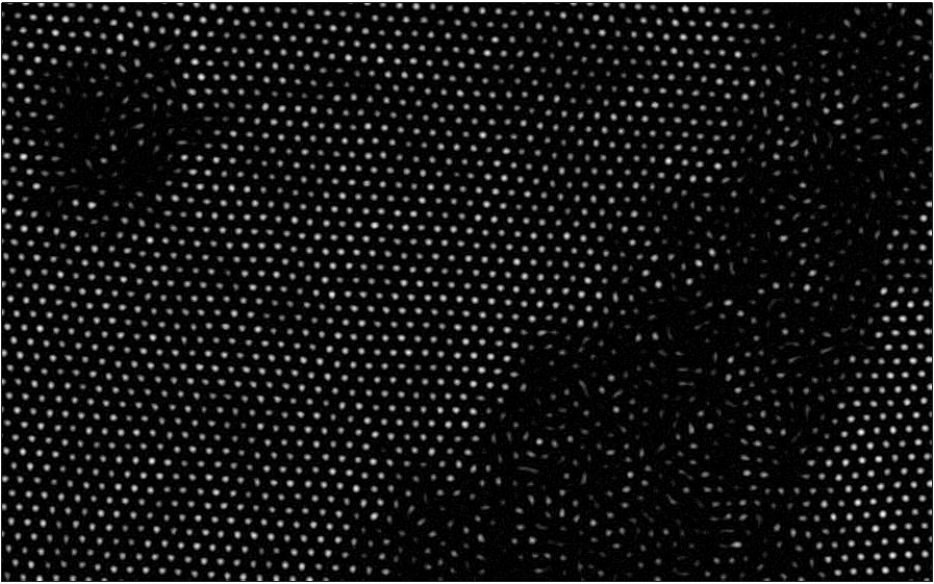}& \includegraphics[height=1.2in]{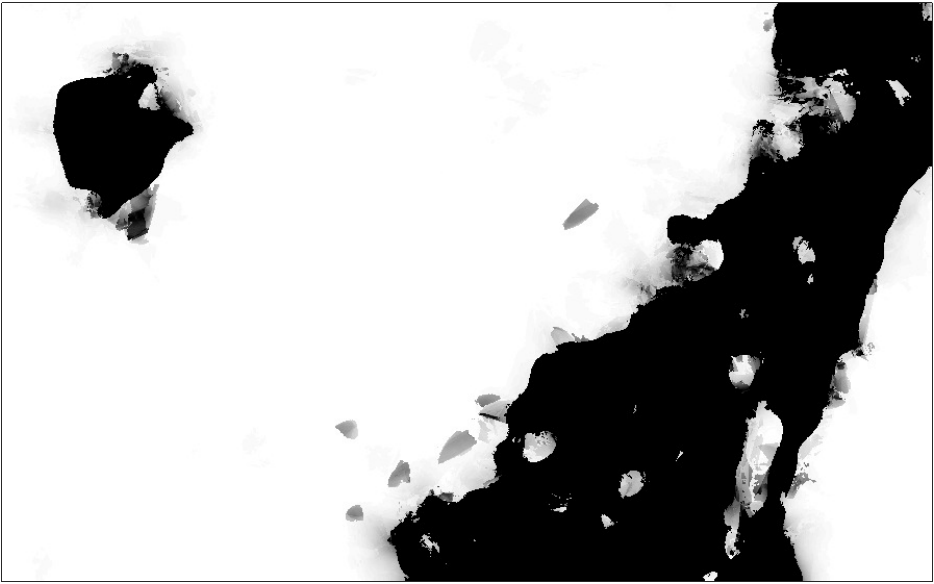} 
             \end{tabular}
  \end{center}
  \caption{Left: melting behavior of thin crystalline films \cite{melting5}. Courtesy of Yilong Han of Hong Kong University of Science and Technology, and Arjun Yodh of Pennsylvania State University  \cite{melting2}. Right: defect and liquid regions are indicated in black, while solid crystal regions and grain islands are in white.}
  \label{fig:ex6}
\end{figure}

 \begin{figure}
  \begin{center}
  \begin{tabular}{cc}
        \includegraphics[height=1.2in]{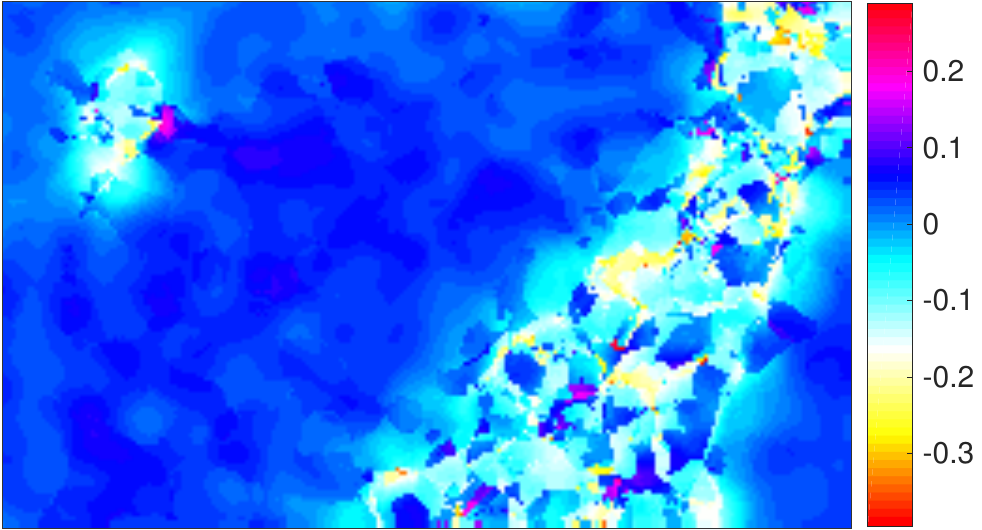} & \includegraphics[height=1.2in]{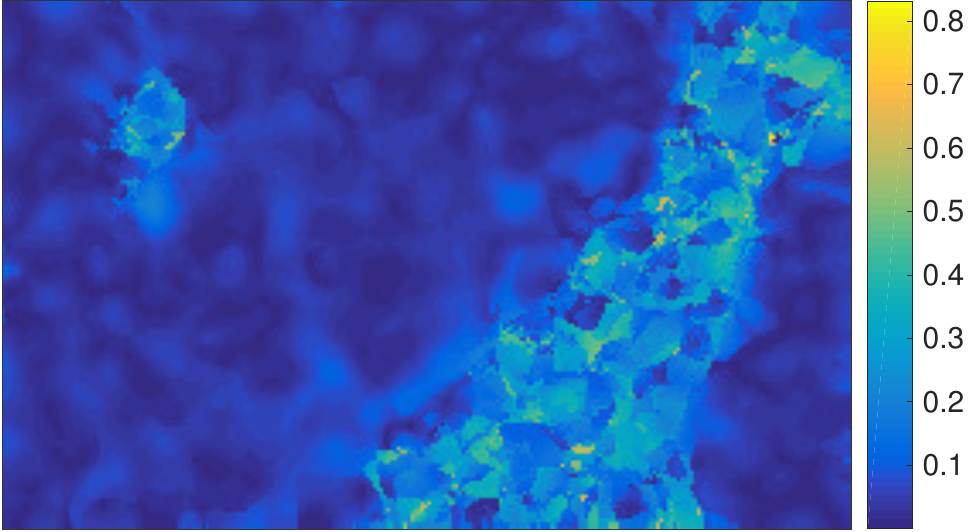} 
     \end{tabular}
  \end{center}
  \caption{Deformation analysis of the example in Figure \ref{fig:ex6} (left). Left: the volume distortion of the elastic deformation. Right: the difference in principle stretches of the elastic deformation. }
  \label{fig:rs6}
\end{figure}

Crystal melting (solid-liquid phase transition) is of considerable importance, but our understanding of the melting process in the atomic scale is far from complete. In particular, the kinetics of this phase transition have proved difficult to predict \cite{Ari}. Scientists have been trying to verify old conjectures and establish new theories in describing the crystal melting process at the atomic scale \cite{melting2,melting5}. 

Figure \ref{fig:ex6} shows an example of thin crystalline films during the melting process \cite{melting5}. We apply the proposed algorithms in this paper to identify the solid and liquid regions, and estimate the elastic deformation of this example. As shown in Figure \ref{fig:ex6} (right), the estimation of the interfaces between solid and liquid states are in line with the image in Figure \ref{fig:ex6} (left) by visual inspection. This result also matches the fact that capillary waves roughen the solid-liquid surface, but locally the intrinsic interface is sharply defined \cite{hernandez2009equilibrium}. Figure \ref{fig:rs6} shows the estimation of the crystal orientation, the volume distortion and the difference in principle stretches of the elastic deformation (only the results in the solid part are informative). These results quantitatively show the interaction between the solid and liquid parts. The elastic deformation of the solid crystal structure near the solid-liquid interfaces has a larger strain.

\subsubsection{Solid-solid phase transition}

It is well known that different geometric arrangements of the same atom or crystal phase can produce materials with different properties. Solid-solid phase transitions can significantly change the physical properties of crystalline solids. A spectacular example of this effect is coal and diamond. Scientists have been trying to identify the right circumstances under which the phase transition occurs, to understand the mechanisms that facilitate phase transitions, and  to control the transition process \cite{Peng:2014}.

 \begin{figure}
  \begin{center}
  \begin{tabular}{ccc}
        \includegraphics[height=1.2in]{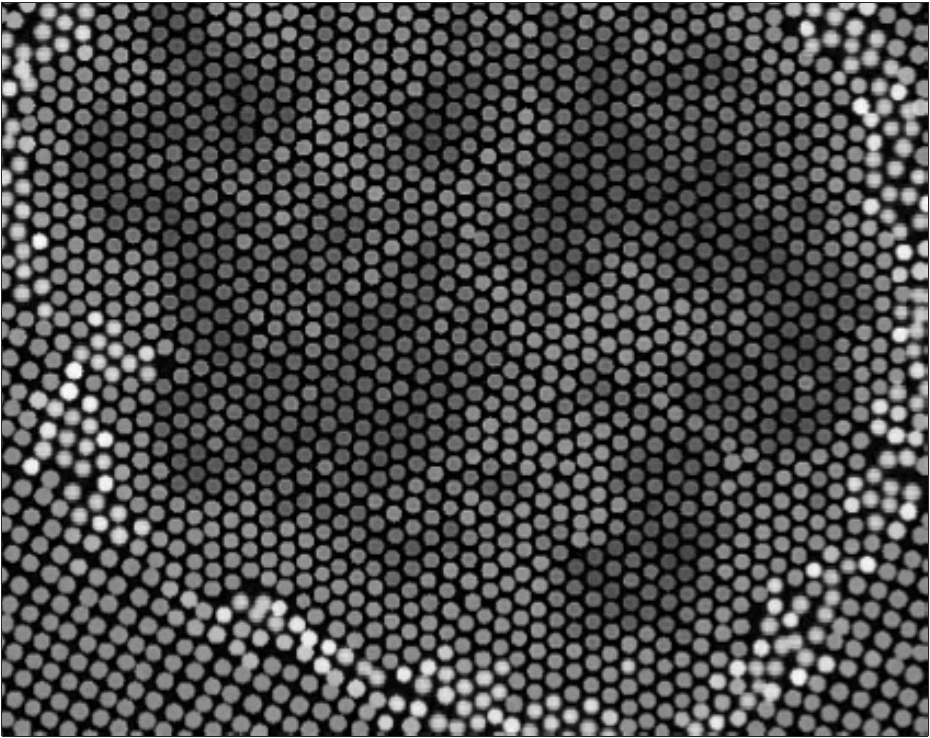}& \includegraphics[height=1.2in]{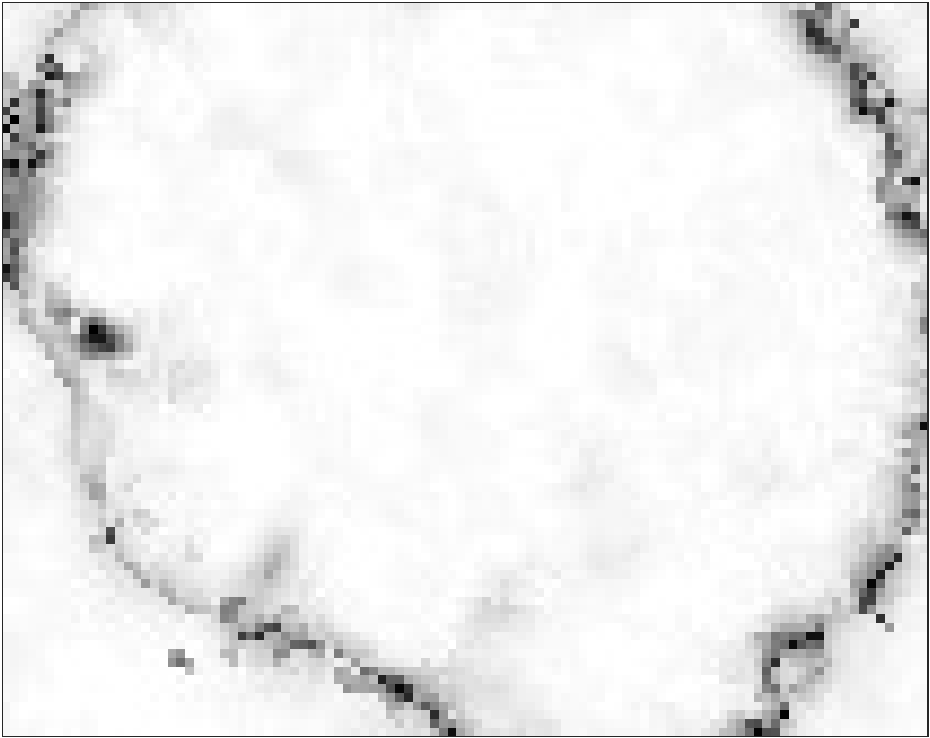} & \includegraphics[height=1.2in]{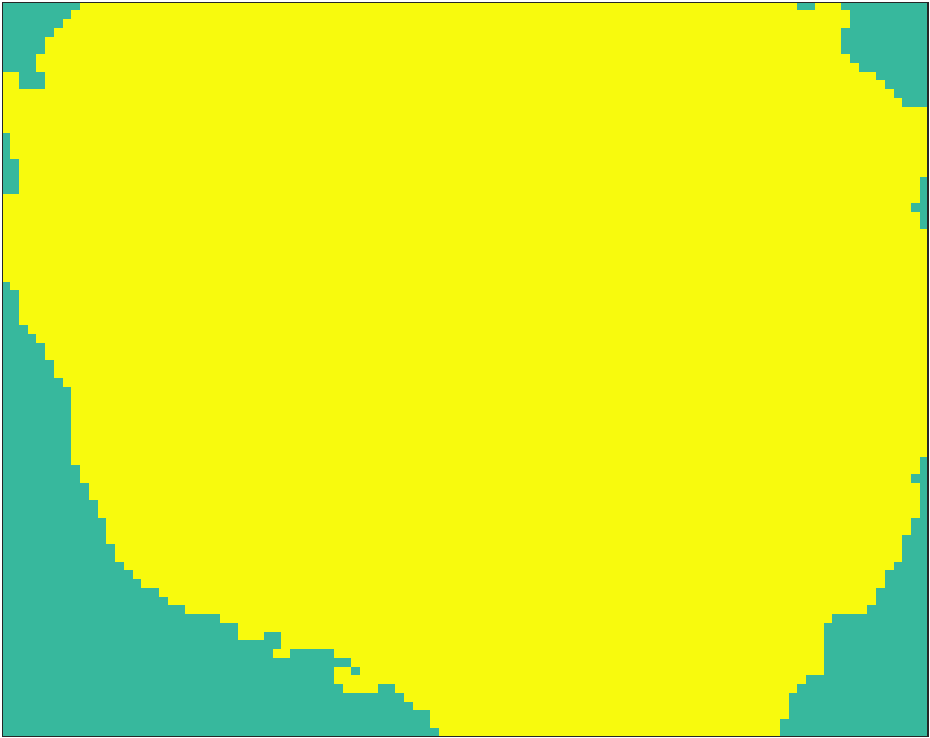} 
     \end{tabular}
  \end{center}
  \caption{Left: a microscopic image showing the simulation of nucleation mechanism in solid-solid phase transitions \cite{Peng:2014}. Courtesy of Yilong Han of Hong Kong University of Science and Technology, and Arjun Yodh of Pennsylvania State University. Middle: defects and interfaces are visualized in black. Right: identified hexagonal lattice in yellow and square lattice in green.}
  \label{fig:ex9}
\end{figure}

Figure \ref{fig:ex9} (top-left) shows the simulation of nucleation mechanism in solid-solid phase transitions \cite{Peng:2014}. In this example, we focus on identifying different crystal patterns for the sake of shortening the manuscript. We apply the proposed algorithms in this paper to identify two solid crystal patterns, and defects. As shown in Figure \ref{fig:ex9} (right two figures), the estimation of the interface between two different solid crystal patterns agrees with the image in Figure \ref{fig:ex9} (top-left) by visual inspection.

\subsection{Self-assembly}
A disordered system of pre-existing components can form an organized structure or pattern through local interactions among the components themselves without external direction. This evolution process is called self-assembly. Self-assembly is an attractive approach for fabricating complex synthetic structures with specific functions \cite{assembly1}. The most common approach to design a self-assembly strategy is by trial and error, where various synthesis methods are studied. 
To discover new self-assembly strategies to make artificial materials with desired mechanical and biological properties, it is crucial to understand the detailed dynamics of formation in self-assembly.


 \begin{figure}
  \begin{center}
  \begin{tabular}{cc}
        \includegraphics[height=1.7in]{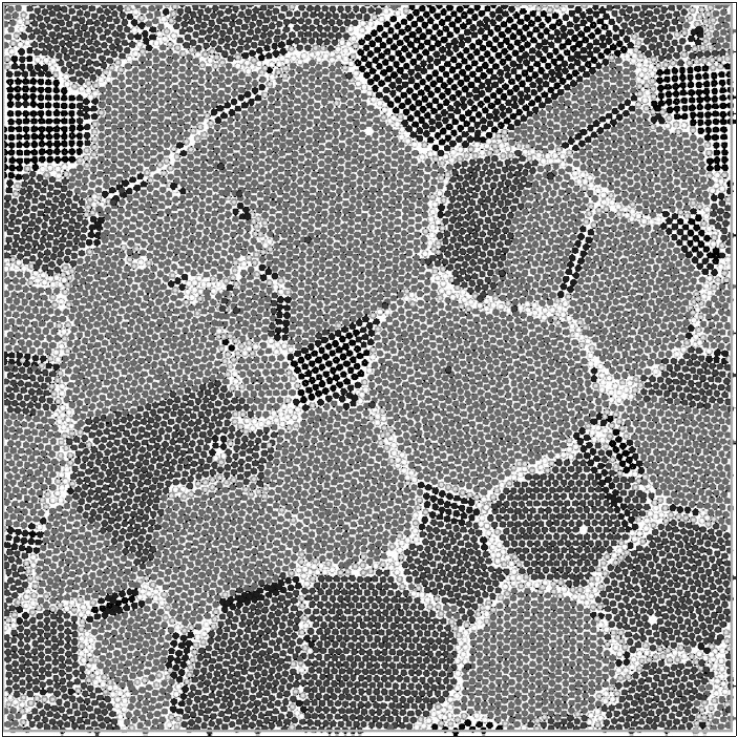}& \includegraphics[height=1.7in]{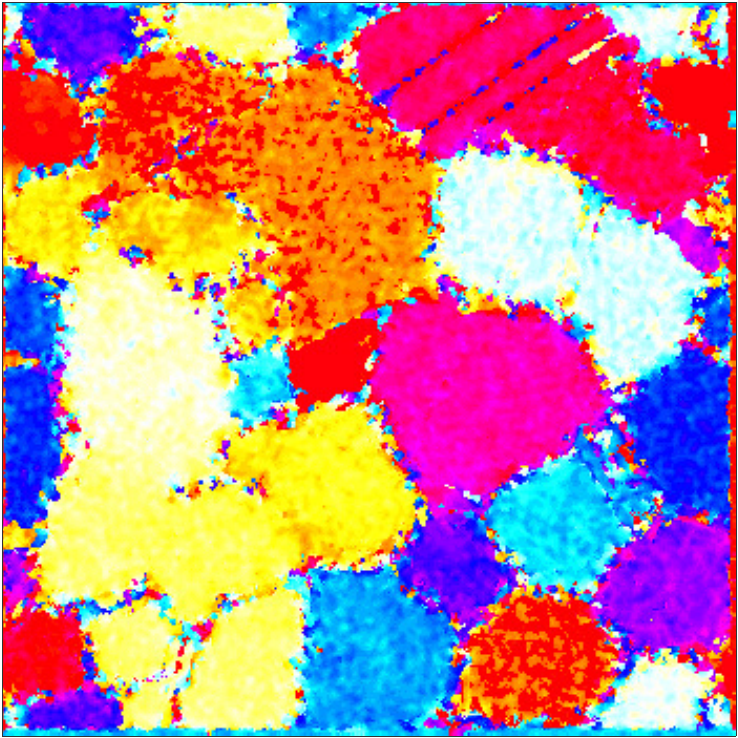} \\
        \includegraphics[height=1.7in]{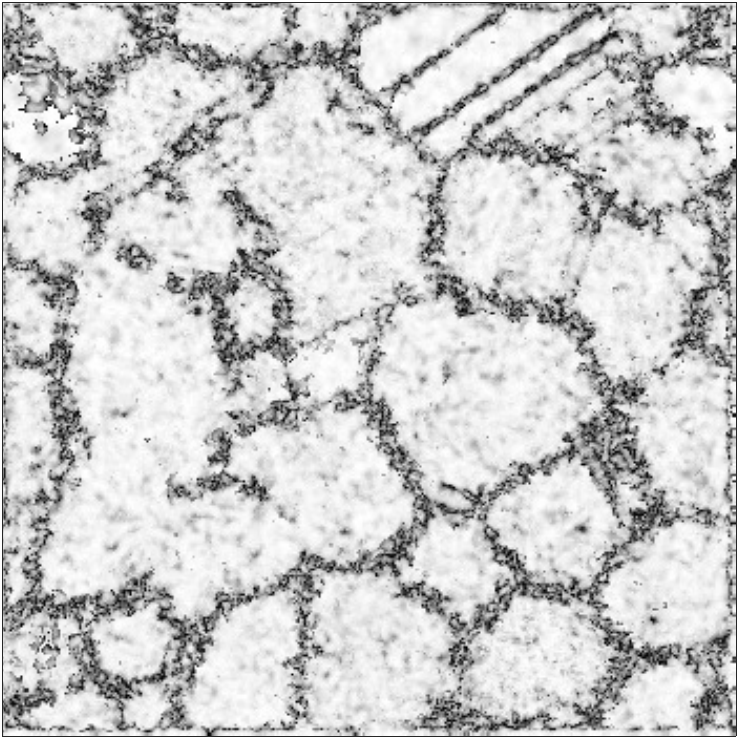}& \includegraphics[height=1.7in]{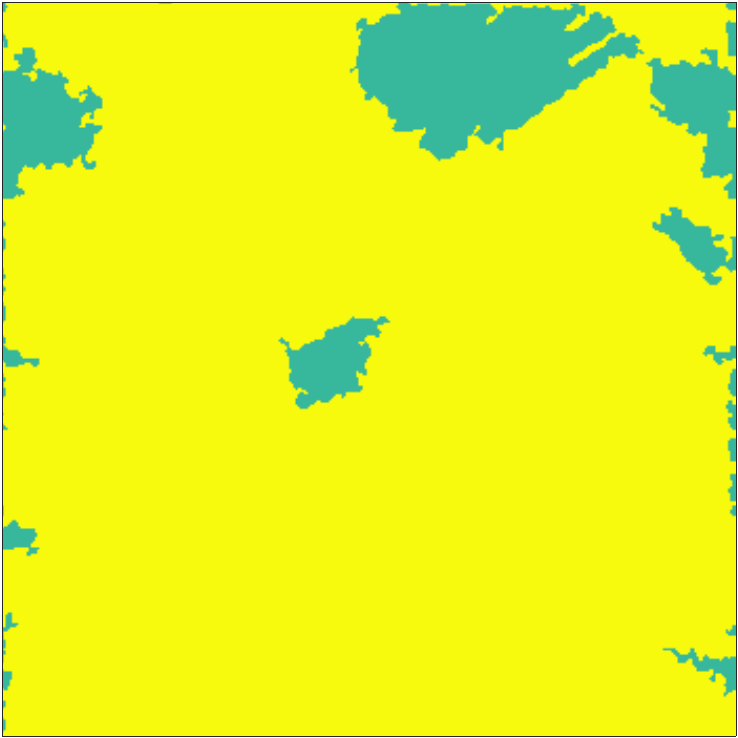} 
     \end{tabular}
  \end{center}
  \caption{Top-left: an example of the crystallization dynamics of sedimenting hard spheres in large systems \cite{Marechal:2011}. Courtesy of Marjolein Dijkstra at Utrecht University. Top-right: different colors encode different crystal orientations. Bottom-left: identified defects in black. Bottom-right: identified square crystal lattice in green and hexagonal crystal lattice in yellow.}
  \label{fig:ex10}
\end{figure}


Figure \ref{fig:ex10} shows the final configuration of the crystallization dynamics of sedimenting hard spheres \cite{Marechal:2011}. The crystallization process is a purely entropy-driven phase transition from a disordered fluid phase to face-centered-cubic crystal structures and hexagonal-close-packed structures. We apply the proposed method in this paper to analyze the crystal image in Figure \ref{fig:ex10} (top-left). Numerical results of crystal orientations, defects, and types of crystal patterns are visualized in Figure \ref{fig:ex10} top-right, bottom-left, and bottom right, respectively. Although the proposed method would misclassify thin grains with diameter less than four atoms, due to the limitation by Heisenberg uncertainty principle since the proposed method is based on phase space analysis, crystal types of larger grains can still be correctly identified. 

\subsection{Out-of-focus issue}

In the last numerical example, we test the proposed crystal analysis method on an image with the out-of-focus problem, where it is difficult to determine the positions of atoms and even gain boundaries by visual inspection (see Figure \ref{fig:exout} (left) for examples). Some structured noise in form of dark disks also increases the difficulty of image analysis. We apply the proposed method in this paper to analyze the example in Figure \ref{fig:exout} (left), show the defect estimation in the middle of Figure \ref{fig:exout}, and visualize the crystal orientation estimation in right panel of Figure \ref{fig:exout}. Numerical results show that the proposed method is stable to the out-of-focus problem and the structured noise. For example, there is a dark disk in the middle of the crystal image but the crystal configuration in terms of atom positionsl smoothly changes; correspondingly, the proposed identifies no grain boundary in this area and shows that the orientation changes smoothly. The crystal configuration starting from the dark disk and going towards the top-left corner changes from white atoms to black atoms, and finally white atoms again. The proposed method identifies no grain boundary and visualizes the smooth change of crystal orientation. In the middle-left part, though the crystal configuration is not clear, the proposed method is still able to identify a small piece of grain and shows its boundary.

 \begin{figure}
  \begin{center}
  \begin{tabular}{ccc}
                \includegraphics[height=1.3in]{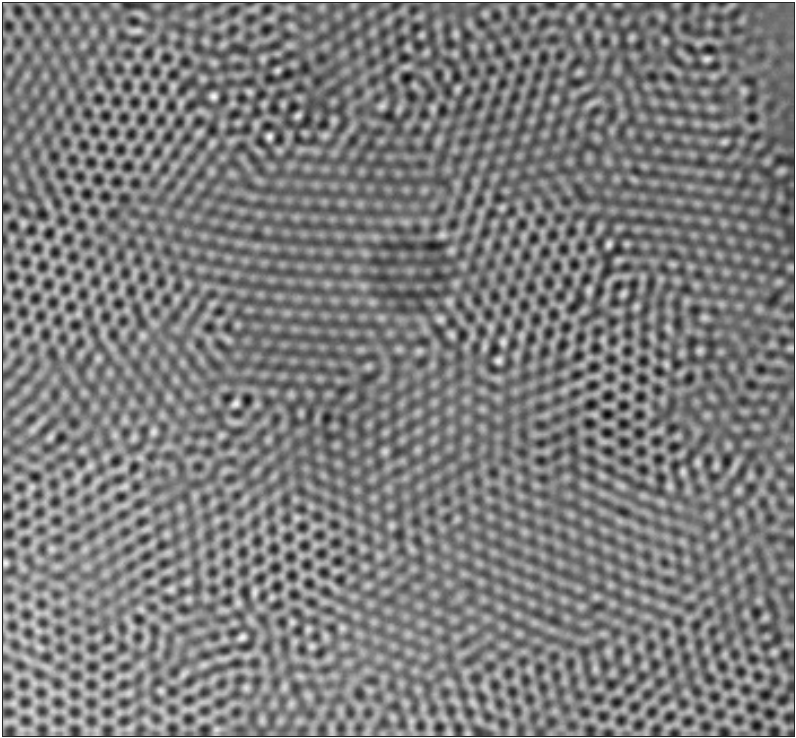}&\includegraphics[height=1.3in]{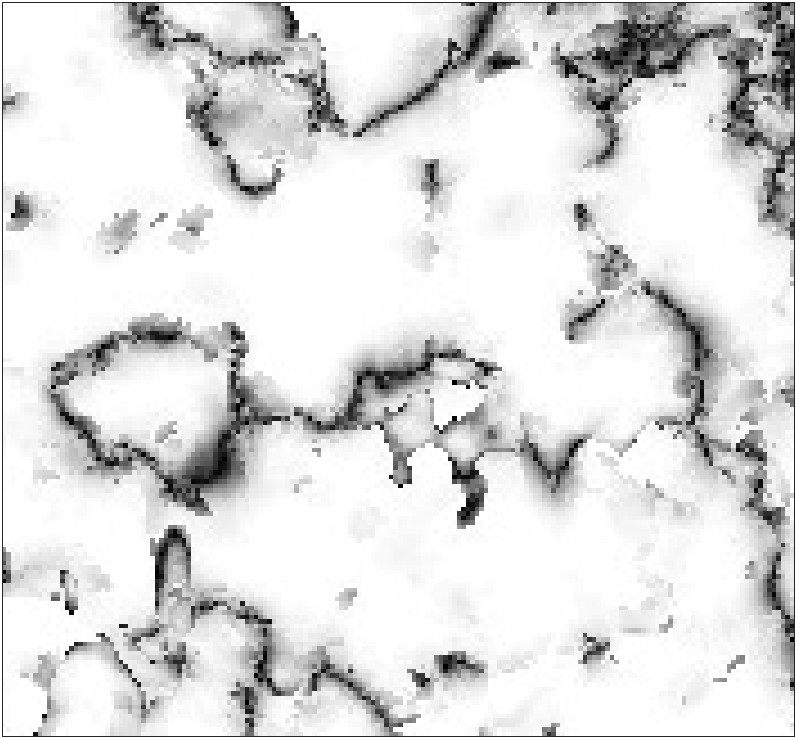}& \includegraphics[height=1.3in]{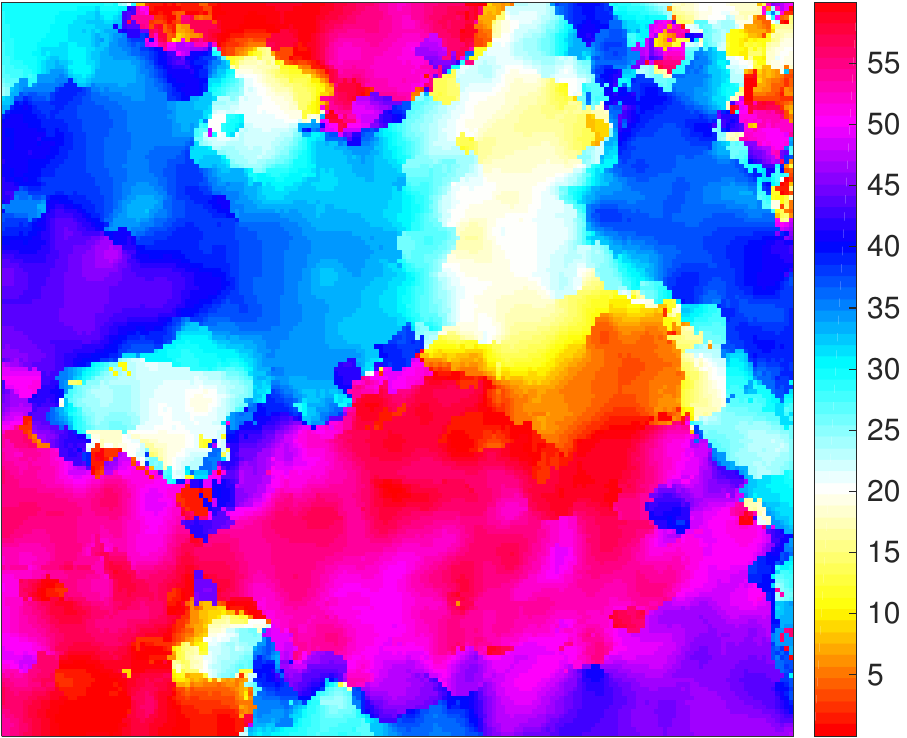}  
     \end{tabular}
  \end{center}
  \caption{Left: an example of atomic
  resolution crystal images with the out-of-focus problem.
  Courtesy of Yilong Han of Hong Kong University of Science and Technology. Middle: identified defects. Right: crystal orientation estimation.}
  \label{fig:exout}
\end{figure}

\section{Conclusion}
\label{sec:conclusion}

We propose a tool set for automatic and quantitative characterization
of complex microscopic crystal images based on the phase-space
sketching and synchrosqueezed transforms. This method encodes
microscopic crystal images into a translation, rotation, illumination,
scale, and small deformation invariant representation. We have applied
this method to analyze various atomic resolution crystal images in
materials science, e.g., crack propagation images, phase transition
images, self-assembly images. Let us mention two possible future
directions: 1) it is worth exploring other advanced image segmentation
algorithms to improve the performance of the proposed method at grain
boundaries; 2) it is useful to extend the current framework for
three-dimensional atomic resolution crystal image analysis.

\bibliographystyle{abbrv}
\bibliography{ref}

\end{document}